\documentclass[10pt,twocolumn]{article}
\usepackage[cmex10]{amsmath}
\usepackage{amsfonts}
\usepackage{amssymb}
\usepackage{amsthm}
\usepackage[pdftex]{graphicx}
\usepackage{subfigure}
\usepackage{array}
\usepackage{cases}
\usepackage{bm}
\usepackage{upgreek}
\usepackage{amsmath}
\usepackage[left=2cm,right=2cm,top=2cm,bottom=2cm]{geometry}
\usepackage{authblk}
\usepackage{setspace}
\usepackage{cite}
\usepackage{enumerate}
\usepackage{algorithmic}
\usepackage{algorithm}
\usepackage{multirow}
\usepackage{graphicx}
\usepackage{booktabs}
\usepackage[font=small,labelfont=bf]{caption}
\usepackage[T1]{fontenc}
\usepackage{pslatex}

\newtheorem{lemma}{Lemma}
\newtheorem{problem}{Problem}
\newtheorem{assumption}{Assumption}
\newtheorem{theorem}{Theorem}
\newtheorem{remark}{Remark}
\newtheorem{definition}{Definition}

\newtheorem{proposition}{Proposition}
\newtheorem{corollary}{Corollary}

\begin{document}

\title{\LARGE \bf Stabilizing Transmission Intervals for Nonlinear Delayed Networked Control Systems\\ \textsc{[Extended Version]}}
\author{Domagoj~Toli\'{c}~and~Sandra~Hirche
\thanks{D. Toli\'{c} is with Faculty of Electrical Engineering and Computing, University of Zagreb, Unska 3, 10000 Zagreb, Croatia. S. Hirche is with the Chair of Information-oriented Control, Technical University of Munich, Arcisstra{\ss}e 21, D-80290 Munich, Germany. (e-mail: domagoj.tolic@fer.hr, hirche@tum.de).}
\thanks{This work has been supported by the European Community Seventh Framework Programme under grant No. 285939 (ACROSS)  and by the Institute for Advanced Study of Technical University of Munich, Germany. D. Toli\'{c} is partially supported by the Air Force Research Laboratory under agreement number FA8655-13-1-3055.}
\thanks{The U.S. Government is authorized to reproduce and distribute reprints for Governmental purposes notwithstanding any copyright notation thereon. The views and conclusions contained herein are those of the authors and should not be interpreted as necessarily representing the official policies or endorsements, either expressed or implied, of the Air Force Research Laboratory or the U.S. Government.}
}

\date{}
\maketitle

\begin{abstract}
\textbf{\small{In this article, we consider a nonlinear process with delayed dynamics to be controlled over a communication network in the presence of disturbances and study robustness of the resulting closed-loop system with respect to network-induced phenomena such as sampled, distorted, delayed and lossy data as well as scheduling protocols. For given plant-controller dynamics and communication network properties (e.g., propagation delays and scheduling protocols), we quantify the control performance level (in terms of $\mathcal{L}_p$-gains) as the transmission interval varies. Maximally Allowable Transfer Interval (MATI) labels the greatest transmission interval for which a prescribed $\mathcal{L}_p$-gain is attained. The proposed methodology combines impulsive delayed system modeling with Lyapunov-Razumikhin techniques to allow for MATIs that are smaller than the communication delays. Other salient features of our methodology are the consideration of variable delays, corrupted data and employment of model-based estimators to prolong MATIs. The present stability results are provided for the class of Uniformly Globally Exponentially Stable (UGES) scheduling protocols. The well-known Round Robin (RR) and Try-Once-Discard (TOD) protocols are examples of UGES protocols. Finally, two numerical examples are provided to demonstrate the benefits of the proposed approach.}}
\end{abstract}

\section{Introduction}		\label{sec:intro} 

Networked Control Systems (NCSs) are spatially distributed systems for which the communication between sensors, actuators and controllers is realized by a shared (wired or wireless) communication network \cite{IEEEhowto:jhespanha}. NCSs offer several advantages, such as reduced installation and maintenance costs as well as greater flexibility, over conventional control systems in which parts of control loops exchange information via dedicated point-to-point connections. At the same time, NCSs generate imperfections (such as sampled, corrupted, delayed and lossy data) that impair the control system performance and can even lead to instability. In order to reduce data loss (i.e., packet collisions) among uncoordinated NCS links, scheduling protocols are employed to govern the communication medium access. Since the aforementioned network-induced phenomena occur simultaneously, the investigation of their cumulative adverse effects on the NCS performance is of particular interest. This investigation opens the door to various trade-offs while designing NCSs. For instance, dynamic scheduling protocols (refer to \cite{wpmhheemels2010} and \cite{mmamduhiCDC2014}), model-based estimators \cite{testrada2010} or smaller transmission intervals can compensate for greater delays at the expense of increased implementation complexity/costs~\cite{dchristmann2014}.

In this article, we consider a nonlinear delayed system to be controlled by a nonlinear delayed dynamic controller over a communication network in the presence of exogenous/modeling disturbances, scheduling protocols among lossy NCS links, time-varying signal delays, time-varying transmission intervals and distorted data. Notice that networked control is not the only source of delays and that delays might be present in the plant and controller dynamics as well. Therefore, we use the term delayed NCSs. The present article takes up the emulation-based approach from \cite{dtolicCDC2014} for investigating the cumulative adverse effects in NCSs and extends it towards plants and controllers with delayed dynamics as well as towards nonuniform time-varying NCS link delays. In other words, different NCS links induce different and nonconstant delays. It is worth mentioning that \cite{dtolicCDC2014} generalizes \cite{wpmhheemels2010} towards corrupted data and the so-called \textit{large delays}. Basically, we allow communication delays to be larger than the transmission intervals. To the best of our knowledge, the work presented herein is the most comprehensive study of the aforementioned cumulative effects as far as the actual plant-controller dynamics (i.e., time-varying, nonlinear, delayed and with disturbances) and interconnection (i.e., output feedback) as well as the variety of scheduling protocols (i.e., UGES protocols) and other network-induced phenomena are concerned (i.e., variable delays, lossy communication channels with distortions). For instance, \cite{fmazenc2013} focuses on time-varying nonlinear control affine plants (i.e., no delayed dynamics in the plant nor controller) and state feedback with a constant delay whilst neither exogenous/modeling disturbances, distorted data nor scheduling protocols are taken into account. The authors in \cite{akruszewski2012} and \cite{bzhang2014} consider linear control systems, impose Zero-Order-Hold (ZOH) sampling and do not consider noisy data nor scheduling protocols. In addition, \cite{akruszewski2012} does not take into account disturbances. Similar comparisons can be drawn with respect to other related works (see \cite{gcwalsh2002,IEEEhowto:jhespanha,wpmhheemels2010,dtolicCDC2014,fmazenc2013} and the references therein).

In order to account for large delays, our methodology employs impulsive delayed system modeling and Lyapunov-Razumikhin techniques when computing Maximally Allowable Transmission Intervals (MATIs) that provably stabilize NCSs for the class of Uniformly Globally Exponentially Stable (UGES) scheduling protocols (to be defined later on). Besides MATIs that merely stabilize NCSs, our methodology is also capable to design MATIs that yield a prespecified level of control system performance. As in \cite{wpmhheemels2010}, the performance level is quantified by means of $\mathcal{L}_p$-gains. According to the batch reactor case study provided in \cite{dtolicCDC2014}, MATI conservativeness repercussions of our approach for the small delay case appear to be modest in comparison with~\cite{wpmhheemels2010}. This conservativeness emanates from the complexity of the tools for computing $\mathcal{L}_p$-gains of delayed (impulsive) systems as pointed out in Section~\ref{sec:example} and, among others, \cite{dpborgers2014}. On the other hand, delayed system modeling (rather than ODE modeling as in~\cite{wpmhheemels2010}) allows for the employment of model-based estimators, which in turn increases MATIs (see Section~\ref{sec:example} for more). In addition, real-life applications are characterized by corrupted data due to, among others, measurement noise and communication channel distortions. In order to include distorted information (in addition to exogenous/modeling disturbances) into the stability analyses, we propose the notion of $\mathcal{L}_p$-stability with bias.

The main contributions of this article are fourfold: a) the design of MATIs in nonlinear delayed NCSs with UGES protocols even for the so-called large delays; b) the Lyapunov-Razumikhin-based procedure for rendering $\mathcal{L}_p$-stability of nonlinear impulsive delayed systems and computing the associated $\mathcal{L}_p$-gains; c) the consideration of NCS links with nonidentical time-dependent delays; and d) the inclusion of model-based estimators. In contrast to our conference paper \cite{dtolicCDC2014}, this article incorporates variable delays, contains proofs, provides a nonlinear numerical example with delayed plant dynamics and designs model-based estimation that prolongs MATIs \cite{testrada2010}. Furthermore, this article accompanies \cite{dtolicTAC2017}.

The remainder of this article is organized as follows. Section~\ref{sec:preliminaries} presents the utilized notation and stability notions regarding impulsive delayed systems. Section~\ref{sec:problem} states the problem of finding MATIs for nonlinear delayed NCSs with UGES protocols in the presence of nonuniform communication delays and exogenous/modeling disturbances. A methodology to solve the problem is presented in Section~\ref{sec:methodology}. Detailed numerical examples are provided in Section~\ref{sec:example}. Conclusions and future challenges are in Section~\ref{sec:concl}. The proofs are provided in the Appendix.

\section{Preliminaries}		\label{sec:preliminaries}

\subsection{Notation}		\label{subsec:notation}

To simplify notation, we use $(x,y) := [x^{\top} \quad y^{\top}]^{\top}$. The dimension of a vector $x$ is denoted $n_x$. Next, let $f : \mathbb{R} \rightarrow \mathbb{R}^n$ be a Lebesgue measurable function on $[a,b] \subset \mathbb{R}$. We use
\begin{align}
	\| f[a,b] \|_p := \left( \int_{\substack{[a,b]}} \|f(s)\|^p \mathrm{d}s	\right)^{\frac{1}{p}}	\nonumber
\end{align}
to denote the $\mathcal{L}_p$-norm of $f$ when restricted to the interval $[a,b]$. If the corresponding norm is finite, we write $f \in \mathcal{L}_p [a,b]$. In the above expression, $\| \cdot \|$ refers to the Euclidean norm of a vector. If the argument of $\| \cdot \|$ is a matrix $A$, then it denotes the induced 2-norm of $A$. Furthermore, $|\cdot|$ denotes the (scalar) absolute value function. The $n$-dimensional vector with all zero entries is denoted $\textbf{0}_n$. Likewise, the $n$ by $m$ matrix with all zero entries is denoted $\textbf{0}_{n \times m}$. The identity matrix of dimension $n$ is denoted $I_n$. In addition, $\mathbb{R}^n_+$ denotes the nonnegative orthant. The natural numbers are denoted $\mathbb{N}$ or $\mathbb{N}_0$ when zero is included.

Left-hand and right-hand limits are denoted $x(t^-)=\lim_{t' \nearrow t} x(t')$ and $x(t^+)=\lim_{t' \searrow t} x(t')$, respectively. Next, for a set $\mathcal{S} \subseteq \mathbb{R}^n$, let $PC([a, b], \mathcal{S}) = \big\{\phi: [a, b] \rightarrow \mathcal{S} \; \big| \; \phi(t)=\phi(t^+)$ for every $t \in [a, b), \; \phi(t^-)$ exists in $\mathcal{S}$ for all $t \in (a, b]$ and $\phi(t^-)=\phi(t)$ for all but at most a finite number of points $t \in (a, b] \big\}$. Observe that $PC([a, b], \mathcal{S})$ denotes the family of right-continuous functions on $[a,b)$ with finite left-hand limits on $(a,b]$ contained in $S$ and whose discontinuities do not accumulate in finite time. Finally, let $\tilde{\textbf{0}}_{n_x}$ denote the zero element of $PC([-d, 0], \mathbb{R}^{n_x})$.

\subsection{Impulsive Delayed Systems}	\label{subsec:hybrid systems with switches}

In this article, we consider nonlinear impulsive delayed systems
\begin{align}	\label{eq:general hybrid sys}
\Sigma \left\{
		\begin{aligned}
			  &\chi(t^+) = h_{\chi}(t,\chi_t) \;\;\;\;\;\; \qquad t \in \mathcal{T}\\
			&\left.
			\begin{aligned}
				\; \dot{\chi}(t) &= f_{\chi}(t,\chi_t,\omega) \\
				\; y &= \ell_{\chi}(t,\chi_t,\omega)
			\end{aligned} \;\;	\right\} \quad \mbox{otherwise}\ ,
		\end{aligned}	\right.
\end{align}
where $\chi \in \mathbb{R}^{n_{\chi}}$ is the state, $\omega \in \mathbb{R}^{n_{\omega}}$ is the input and $y \in \mathbb{R}^{n_{y}}$ is the output. The functions $f_{\chi}$ and $h_{\chi}$ are regular enough to guarantee forward completeness of solutions which, given initial time $t_{0}$ and initial condition $\chi_{t_0} \in PC([-d,0], \mathbb{R}^{n_{\chi}})$, where $d \geq 0$ is the maximum value of all time-varying delay phenomena, are given by right-continuous functions $t \mapsto \chi(t) \in PC([t_0-d, \infty], \mathbb{R}^{n_{\chi}})$. Furthermore, $\chi_t$ denotes the translation operator acting on the trajectory $\chi(\cdot)$ defined by $\chi_t(\theta):=\chi(t+\theta)$ for $-d \leq \theta \leq 0$. In other words, $\chi_t$ is the restriction of trajectory $\chi(\cdot)$ to the interval $[t-d,t]$ and translated to $[-d,0]$. For $\chi_t \in PC([-d,0],\mathbb{R}^{n_{\chi}})$, the norm of $\chi_t$ is defined by $\|\chi_t\| =\sup_{-d \leq \theta \leq 0} \|\chi_t(\theta)\|$. Jumps of the state are denoted $\chi(t^+)$ and occur at time instants $t \in \mathcal{T} := \{t_1,t_2,\ldots \}$, where $t_i < t_{i+1}$, $i \in \mathbb{N}_0$. The value of the state after a jump is given by $\chi(t^+)$ for each $t \in \mathcal{T}$. For a comprehensive discussion regarding the solutions to (\ref{eq:general hybrid sys}) considered herein, refer to \cite[Chapter 2 \& 3]{ghballinger1999}. Even though the considered solutions to (\ref{eq:general hybrid sys}) allow for jumps at $t_0$, we exclude such jumps in favor of notational convenience.

\begin{definition}	[Uniform Global Stability]	\label{def:UGS}
For $\omega \equiv \textbf{0}_{n_{\omega}}$, the system $\Sigma$ is said to be Uniformly Globally Stable (UGS) if for any $\epsilon > 0$ there exists $\delta(\epsilon) > 0$ such that, for each $t_{0} \in \mathbb{R}$ and each $\chi_{t_0} \in PC([-d,0], \mathbb{R}^{n_{\chi}})$ satisfying $\|\chi_{t_0}\| < \delta(\epsilon)$, each solution $t \mapsto \chi(t) \in PC([t_0-d, \infty], \mathbb{R}^{n_{\chi}})$ to $\Sigma$ satisfies $\| \chi(t) \| < \epsilon$ for all $t \geq t_0$ and $\delta(\epsilon)$ can be chosen such that $\lim_{\epsilon \rightarrow \infty} \delta(\epsilon)~=~\infty$.
\end{definition}

\begin{definition} [Uniform Global Asymptotic Stability]	\label{def:UGAS}
For $\omega \equiv \textbf{0}_{n_{\omega}}$, the system $\Sigma$ is said to be Uniformly Globally Asymptotically Stable (UGAS) if it is UGS and uniformly globally attractive, i.e., for each $\eta,\zeta>0$ there exists $T(\eta,\zeta) > 0$ such that $\| \chi(t) \| < \eta$ for every $t \geq t_0 + T(\eta,\zeta)$ and every $\| \chi_{t_0} \| < \zeta$.
\end{definition}

\begin{definition}	[Uniform Global Exponential Stability] \label{def:UGES}
For $\omega \equiv \textbf{0}_{n_{\omega}}$, the system $\Sigma$ is said to be Uniformly Globally Exponentially Stable (UGES) if there exist positive constants $\lambda$ and $M$ such that, for each $t_{0} \in \mathbb{R}$ and each $\chi_{t_0} \in PC([-d,0], \mathbb{R}^{n_{\chi}})$, each solution $t \mapsto \chi(t) \in PC([t_0-d, \infty], \mathbb{R}^{n_{\chi}})$ to $\Sigma$ satisfies $\| \chi(t) \| \leq M \|\chi_{t_0}\| e^{-\lambda(t-t_0)}$ for each $t \geq t_0$.
\end{definition}

\begin{definition} [$\mathcal{L}_p$-Stability with Bias $b$]	\label{def:stability bias}
Let $p \in [1,\infty]$. The system $\Sigma$ is $\mathcal{L}_p$-stable with bias $b(t) \equiv b \geq 0$ from $\omega$ to $y$ with (linear) gain $\gamma \geq 0$ if there exists $K \geq 0$ such that, for each $t_{0} \in \mathbb{R}$ and each $\chi_{t_0} \in PC([-d,0], \mathbb{R}^{n_{\chi}})$, each solution to $\Sigma$ from $\chi_{t_0}$ satisfies $\|y[t_{0},t]\|_p \leq K \|\chi_{t_0}\| + \gamma \|\omega[t_{0},t]\|_p + \|b[t_0,t]\|_p$ for each $t \geq t_{0}$.
\end{definition}

\begin{definition} [$\mathcal{L}_p$-Detectability]	\label{def:detectability}
Let $p \in [1,\infty]$. The state $\chi$ of $\Sigma$ is $\mathcal{L}_p$-detectable from $(y, \omega)$ with (linear) gain $\gamma \geq 0$ if there exists $K \geq 0$ such that, for each $t_{0} \in \mathbb{R}$ and each $\chi_{t_0} \in PC([-d,0], \mathbb{R}^{n_{\chi}})$, each solution to $\Sigma$ from $\chi_{t_0}$ satisfies $\|\chi[t_{0},t]\|_p \leq K \|\chi_{t_0}\| + \gamma \|y[t_{0},t]\|_p + \gamma \|\omega[t_{0},t]\|_p$ for each $t \geq t_{0}$.
\end{definition}

\noindent
Definitions \ref{def:UGS}, \ref{def:UGAS} and \ref{def:UGES} are motivated by \cite{IEEEhowto:hkhalil}, while Definition \ref{def:detectability} is inspired by \cite{dnesic2004}. Definition \ref{def:stability bias} is motivated by \cite{dnesic2004} and \cite{zjiang1994}. When $b=0$, we say ``$\mathcal{L}_p$-stability'' instead of ``$\mathcal{L}_p$-stability with bias $0$''.

\section{Problem Formulation}		\label{sec:problem}

Consider a nonlinear control system consisting of a plant with delayed dynamics
\begin{align}
	\dot{x}_p &= f_p (t,x_{p_t},u,\omega_p),		\nonumber	\\
	y &= g_p (t,x_{p_t}),	\label{eq:plant}
\end{align}
and a controller with delayed dynamics
\begin{align}
	\dot{x}_c &= f_c (t,x_{c_t},y,\omega_c),	\nonumber \\
	u &= g_c (t,x_{c_t})	\label{eq:controller},
\end{align}
where $x_p \in \mathbb{R}^{n_p}$ and $x_c \in \mathbb{R}^{n_c}$ are the states, $y \in \mathbb{R}^{n_y}$ and $u \in \mathbb{R}^{n_u}$ are the outputs, and $(u,\omega_p) \in \mathbb{R}^{n_u} \times \mathbb{R}^{n_{\omega_p}}$ and $(y,\omega_c) \in \mathbb{R}^{n_y} \times \mathbb{R}^{n_{\omega_c}}$ are the inputs of the plant and controller, respectively, where $\omega_p$ and $\omega_c$ are external disturbances to (and/or modeling uncertainties of) the plant and controller, respectively. The translation operators $x_{p_t}$ and $x_{c_t}$ are defined in Section \ref{subsec:hybrid systems with switches} while the corresponding plant and controller delays are $d_p \geq 0$ and $d_c \geq 0$, respectively. For notational convenience, constant plant and controller delays are considered.

Let us now model the communication network between the plant and controller over which intermittent and realistic exchange of information takes place (see Figure~\ref{fig:NCS}). The value of $u$ computed by the controller that arrives at the plant is denoted $\hat{u}$. Similarly, the values of $y$ that the controller actually receives are denoted $\hat{y}$. Consequently, we have 
\begin{align} \label{eqn:InterconnectionAssignment}
	u = \hat{u}, \qquad y = \hat{y},
\end{align}
on the right hand sides of (\ref{eq:plant}) and (\ref{eq:controller}). In our setting, the quantity $\hat u$ is the delayed and distorted input $u$ fed to the plant~(\ref{eq:plant}) while the quantity $\hat y$ is the delayed and distorted version of $y$ received by the controller~(\ref{eq:controller}). We proceed further by defining the error vector
\begin{equation}
	e = \begin{bmatrix}
				 			e_{y}(t) \\
				 			e_{u}(t)
				 	\end{bmatrix} 
					:=	
	\begin{bmatrix}
							\hat{y}(t) - y_t	\\
							\hat{u}(t) - u_t
			 \end{bmatrix},
 \label{sensing error}
\end{equation}
where $y_t$ and $u_t$ are translation operators and the maximal network-induced delay $d \geq 0$ (e.g., propagation delays and/or delays arising from protocol arbitration). The operator $(y_t,u_t)$ in (\ref{sensing error}) delays each component of $(y,u)$ for the respective delay. Essentially, if the $i^{\mathrm{th}}$ component of $(y(t),u(t))$, that is $(y(t),u(t))_i$, is transmitted with delay $d_i: \mathbb{R} \rightarrow \mathbb{R}_+$, then the $i^{\mathrm{th}}$ component of $(y_t,u_t)$, that is $(y_t,u_t)_i$, is in fact $(y(t-d_i(t)),u(t-d_i(t)))_i$. Accordingly, $d := \max \{\sup_{t \in \mathbb{R}}d_1(t),\ldots,\sup_{t \in \mathbb{R}}d_{n_y+n_u}(t)\}$.

\begin{figure}
  \begin{center}
    \includegraphics[width=0.48\textwidth]{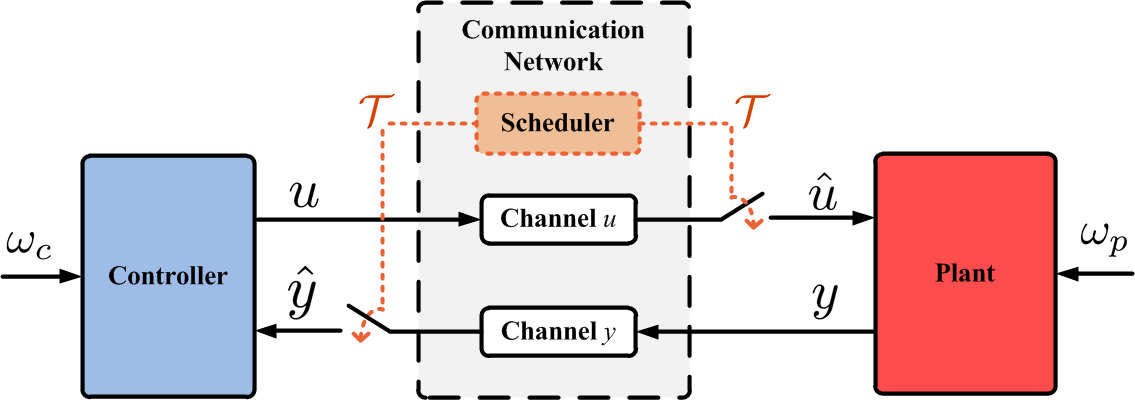}
  \end{center}
  \caption{A diagram of a control system with the plant and controller interacting over a communication network with intermittent information updates. The two switches indicate that the information between the plant and controller are exchanged (complying with some scheduling protocol among the NCS links) at discrete time instants belonging to a set~$\mathcal{T}$. The communication delays in each NCS link are time varying and, in general, different.}
  \label{fig:NCS}
\end{figure}

Due to intermittent transmissions of the components of $y$ and $u$, the respective components of $\hat{y}$ and $\hat{u}$ are updated at time instants $t_1,t_2, \ldots, t_i, \ldots \in \mathcal{T}$, i.e.,
\begin{align}	\label{eq:jump equations}
		\left.
			\begin{aligned}
				\hat{y}(t_i^+) &= y_t + h_y(t_i,e(t_i))	 \\
				\hat{u}(t_i^+) &= u_t + h_u(t_i,e(t_i))
			\end{aligned} \;\;\;	\right\} \;\;\; t_i \in \mathcal{T},
\end{align}	
where $h_y: \mathbb{R} \times \mathbb{R}^{n_e} \rightarrow \mathbb{R}^{n_y}$ and $h_u: \mathbb{R} \times \mathbb{R}^{n_e} \rightarrow \mathbb{R}^{n_u}$ model measurement noise, channel distortion and the underlying scheduling protocol. The role of $h_y$ and $h_u$ is as follows. Suppose that the NCS has $l$ links. Accordingly, the error vector $e$ can be partitioned as $e:=(e_1,\ldots,e_l)$. In order to avoid cumbersome indices, let us assume that each NCS link is characterized by its own delay. Hence, there are merely $l$ (rather than $n_y+n_u$) different delays $d_i: \mathbb{R} \rightarrow \mathbb{R}_+$ in (\ref{sensing error}). Besides the already introduced upper bound $d$ on $d_i(t)$'s, we assume that $d_i(t)$'s are differentiable with bounded $|\dot{d}_i(t)|$. As orchestrated by (\ref{eq:jump equations}), if the $j^{\mathrm{th}}$ NCS link is granted access to the communication medium at some $t_i \in \mathcal{T}$, the corresponding components of $(\hat{y}(t_i),\hat{u}(t_i))$ jump to the received values. It is to be noted that all other components of $(\hat{y}(t_i),\hat{u}(t_i))$ remain unaltered. Consequently, the related components of $e(t_i)$ reset to the noise $\nu_j(t_i)$ present in the received data, i.e.,
\begin{align}
	e_j(t_i^+)=\nu_j(t_i),		\label{eq:noise}
\end{align}
and we assume that 
\begin{align*}
	\sup_{t \in \mathbb{R}, j \in \{1,\ldots,l\}} \|\nu_j(t)\| = K_{\nu}.
\end{align*}
Noise $\nu_j(t_i)$, which is embedded in $h_y$ and $h_u$, models any discrepancy between the received values and their actual values at time $t_i-d_j(t)$ (when the $j^{\mathrm{th}}$ NCS link of $(y(t),u(t))$ was sampled). As already indicated, this discrepancy can be a consequence of measurement noise and channel distortion. We point out that $\nu_j$ has nothing to do with~$\omega_p$ nor~$\omega_c$. Observe that out-of-order packet arrivals, as a consequence of the time-varying delays, are allowed for.

In between transmissions, the values of $\hat{y}$ and $\hat{u}$ need not to be constant as in \cite{wpmhheemels2010}, but can be estimated in order to extend transmission intervals (consult \cite{testrada2010} for more). In other words, for each $t \in [t_0,\infty) \setminus \mathcal{T}$ we have
\begin{align}
	\dot{\hat{y}} &= \hat{f}_p \big( t,x_{p_t},x_{c_t},\hat{y}_t,\hat{u}_t,\omega_p,\omega_c \big),			\nonumber \\
	\dot{\hat{u}} &= \hat{f}_c \big( t,x_{p_t},x_{c_t},\hat{y}_t,\hat{u}_t,\omega_p,\omega_c \big),			\label{eq:no estimation}
\end{align}
where the translation operators $\hat{y}_t$ and $\hat{u}_t$ are with delay~$d$. The commonly used ZOH strategy is characterized by $\dot{\hat{y}} \equiv \textbf{0}_{n_y}$ and $\dot{\hat{u}} \equiv \textbf{0}_{n_u}$.

The following definition of UGES scheduling protocols is extracted from \cite{dnesic2004} and \cite{wpmhheemels2010}.
\begin{definition}	\label{def:UGES pro}
	Consider the noise-free setting, i.e., $K_{\nu}=0$. The protocol given by $h:=(h_y,h_u)$ is UGES if there exists a function $W : \mathbb{N}_0 \times \mathbb{R}^{n_e} \rightarrow \mathbb{R}_+$ such that $W(i,\cdot):\mathbb{R}^{n_e} \rightarrow \mathbb{R}_+$ is locally Lipschitz (and hence almost everywhere differentiable) for every $i \in \mathbb{N}_0$, and if there exist positive constants $\underline{a}$, $\overline{a}$ and $0 \leq \rho < 1$ such that
	\begin{enumerate}[(i)]
		\item $\underline{a}\|e\| \leq W(i,e) \leq \overline{a} \|e\|$, and
		\item $W(i+1,h(t_i,e)) \leq \rho W(i,e)$,
	\end{enumerate}
	for all $(i,e) \in \mathbb{N}_0 \times \mathbb{R}^{n_e}$.
\end{definition}

\noindent
Notice that, even though the delays could result from protocol arbitration, the delays are not a part of the UGES protocol definition \cite{dnesic2004,wpmhheemels2010}. In addition, $\mathcal{T}$ is not a part of the protocol, but rather a consequence, as it is yet to be designed. Commonly used UGES protocols are the Round Robin (RR) and Try-Once-Discard protocol (TOD) (consult \cite{dnesic2004,wpmhheemels2010,dchristmann2014}). The corresponding constants are $\underline{a}_{RR}=1$, $\overline{a}_{RR}=\sqrt{l}$, $\rho_{RR}=\sqrt{(l-1)/l}$ for RR and $\underline{a}_{TOD}=\overline{a}_{TOD}=1$, $\rho_{TOD}=\sqrt{(l-1)/l}$ for TOD. Explicit expressions of the noise-free $h(t,e)$ for RR and TOD are provided in \cite{dnesic2004}, but are not needed in the context of this article.

\noindent
The properties imposed on the NCS in Figure \ref{fig:NCS} are summarized in the following standing assumption.
\begin{assumption}	\label{asmp:standing}
	The jump times of the NCS links at the controller and plant end obey the underlying UGES scheduling protocol (characterized through $h$) and occur at transmission instants belonging to $\mathcal{T}:=\{t_1,t_2, \ldots, t_i, \ldots\}$, where $\varepsilon \leq t_{i+1} - t_{i} \leq \tau$ for each $i \in \mathbb{N}_0$ with $\varepsilon > 0$ arbitrarily small. The received data is corrupted by measurement noise and/or channel distortion (characterized through $h$ as well). In addition, each NCS link is characterized by the network-induced delay $d_i(t)$, $i \in \{1,\ldots,l\}$.
\end{assumption}

\noindent
The existence of a strictly positive $\tau$, and therefore the existence of $\varepsilon > 0$, is demonstrated in Remark \ref{rmk:zeno}.

A typical closed-loop system (\ref{eq:plant})-(\ref{eq:no estimation}) with continuous (yet delayed) information flows in all NCS links might be robustly stable (in the $\mathcal{L}_p$ sense according to (\ref{eq:Lp H})) only for some sets of $d_i(t)$, $i \in \{1,\ldots,l\}$. We refer to the family of such delay sets as the family of \textit{admissible delays} and denote it $\mathcal{D}$. Next, given some admissible delays $d_i(t)$, $i \in \{1,\ldots,l\}$, the maximal $\tau$ which renders $\mathcal{L}_p$-stability (with a desired gain) of the closed-loop system (\ref{eq:plant})-(\ref{eq:no estimation}) is called MATI and is denoted $\overline{\tau}$. We are now ready to state the main problem studied herein.
\begin{problem}		\label{prob:main problem}
		Given admissible delays $d_i(t)$, $i \in \{1,\ldots,l\}$, estimator (\ref{eq:no estimation}) and the UGES protocol of interest, determine the MATI $\overline{\tau}$ to update components of $(\hat{y}$,$\hat{u})$ such that the NCS (\ref{eq:plant})-(\ref{eq:no estimation}) is $\mathcal{L}_p$-stable with bias and a prespecified $\mathcal{L}_p$-gain for some $p \in [1,\infty]$.
\end{problem}

\begin{remark}			\label{rmk:false intuition}
Even though our intuition (together with the case studies provided herein and in \cite{dtolicCDC2014}) suggests that merely ``small enough" delays (including the zero delay) are admissible because the control performance impairs (i.e., the corresponding $\mathcal{L}_p$-gain increases) with increasing delays, this observation does not hold in general \cite{jkhale1977},\cite[Chapter 1.]{siniculescu2001},\cite{hsmith2011}. In fact, ``small" delays may destabilize some systems while ``large" delays might destabilize others. In addition, even a second order system with a single discrete delay might toggle between stability and instability as this delay is being decreased. Clearly, the family $\mathcal{D}$ needs to be specified on a case-by-case basis. Hence, despite the fact that the case studies presented herein and in \cite{dtolicCDC2014} yield MATIs that hold for all smaller time-invariant delays (including the zero delay) than the delays for which these MATIs are computed for, it would be erroneous to infer that this property holds in general.
\end{remark}

\section{Methodology} \label{sec:methodology}

Along the lines of \cite{dnesic2004}, we rewrite the closed-loop system (\ref{eq:plant})-(\ref{eq:no estimation}) in the following form amenable for small-gain theorem (see \cite[Chapter 5]{IEEEhowto:hkhalil}) analyses:
\begin{subequations}	\label{eq:NCS system}
\begin{align}
		&\left.
		\begin{aligned}
		x(t^+) &= x(t)	\\
		e(t^+) &= h(t,e(t))
		\end{aligned}	\;\; \right\} \;\;\; t \in \mathcal{T}		\label{eq:NCS system jump}\\		&\left.
		\begin{aligned}
		\dot{x}  &= f(t,x_t,e,\omega)	\\
		\dot{e}  &= g(t,x_t,e_t,\omega_t)
		\end{aligned} \;\; \right\} \;\; \mbox{ otherwise}, \label{eq:NCS system flow}
\end{align}
\end{subequations}
where $x:=(x_p,x_c)$, $\omega:=(\omega_p,\omega_c)$, and functions $f$, $g$ and $h$ are given by (\ref{eq:nesic 1}) and (\ref{eq:nesic 2}). We assume enough regularity on $f$ and $g$ to guarantee existence of the solutions on the interval of interest \cite[Chapter 3]{ghballinger1999}. Observe that differentiability of $d_i(t)$'s and boundedness of $|\dot{d}_i(t)|$ play an important role in attaining regularity of $g$. For the sake of simplicity, our notation does not explicitly distinguish between translation operators with delays $d_p$, $d_c$, $d$ or $2 d$ in (\ref{eq:nesic 1}), (\ref{eq:nesic 2}) and in what follows. In this regard, we point out that the operators $x_{p_t}$ and $x_{c_t}$ are with delays $d_p$ and $d_c$, respectively, the operators $g_{p_t}$ and $g_{c_t}$ within $\hat{f}_p$ and $\hat{f}_c$ are with delay $2 d$ while all other operators are with delay $d$. In what follows we also use $\overline{d}:=2d+\max\{d_p,d_c\}$, which is the maximum value of all delay phenomena in (\ref{eq:nesic 2}).

\begin{figure*}
\begin{align}
	f&(t,x_t,e,\omega) \overset{\textstyle{\textrm{ (\ref{eq:plant}),(\ref{eq:controller}) }}}{:=}	\begin{bmatrix}
														f_p \big(t,x_{p_t},\underbrace{g_{c_t}(t,x_{c_t})+e_u(t)}_{{\scriptstyle = \hat{u}(t) \textrm{ using (\ref{eq:controller}) and (\ref{sensing error})} }},\omega_p(t)\big)		\\
														f_c \big(t,x_{c_t},\underbrace{g_{p_t}(t,x_{p_t})+e_y(t)}_{{\scriptstyle = \hat{y}(t) \textrm{ using (\ref{eq:plant}) and (\ref{sensing error})} }},\omega_c(t)\big)
												\end{bmatrix} =: \begin{bmatrix}
																							f_1(t,x_t,e,\omega)		\\
																							f_2(t,x_t,e,\omega)
																					\end{bmatrix}; \qquad 	h(t,e(t)) := 	\begin{bmatrix}
																																						h_y \big(t,e(t)\big)			\\
																																						h_u \big(t,e(t)\big)
																																				\end{bmatrix}			\label{eq:nesic 1}		\\
	g&(t,x_t,e_t,\omega_t) \overset{\textstyle{\textrm{ (\ref{sensing error}) }}}{:=}
						\begin{bmatrix}
												\underbrace{
												\begin{smallmatrix}
												\hat{f}_p \big(t,x_{p_t},x_{c_t},g_{p_t}(t,x_{p_t})+e_{y_t},g_{c_t}(t,x_{c_t})+e_{u_t},\omega(t) \big)
												\end{smallmatrix}
												}_{{\scriptstyle \textrm{ model-based estimator (\ref{eq:no estimation})}}} \underbrace{ 
												\begin{smallmatrix}
																		- \big(\frac{\partial g_p}{\partial t}\big)_t  (t,x_{p_t})- \big(\frac{\partial g_p}{\partial x_p}\big)_t (t,x_{p_t}) f_{1_t}(t,x_t,e,\omega)
												\end{smallmatrix}
												}_{ \scriptstyle{= -\dot{y}_t \textrm{ using (\ref{eq:plant}) and (\ref{eq:nesic 1})} } } 	\\
												\overbrace{
												\begin{smallmatrix}
												\hat{f}_c \big(t,x_{p_t},x_{c_t},g_{p_t}(t,x_{p_t})+e_{y_t},g_{c_t}(t,x_{c_t})+e_{u_t},\omega(t) \big)
												\end{smallmatrix}
												} \underbrace{
												\begin{smallmatrix}
												- \big(\frac{\partial g_c}{\partial t}\big)_t (t,x_{c_t})- \big(\frac{\partial g_c}{\partial x_c}\big)_t (t,x_{c_t}) f_{2_t}(t,x_t,e,\omega)
												\end{smallmatrix}
												}_{ \scriptstyle{= -\dot{u}_t \textrm{ using (\ref{eq:controller}) and (\ref{eq:nesic 1})} } }
										 \end{bmatrix}			\label{eq:nesic 2} \\ \midrule	\nonumber
\end{align}
\end{figure*}

For future reference, the delayed dynamics 
\begin{subequations}		\label{eq:Sigman}
\begin{align}
		&\left.
		\begin{aligned}
		x(t^+) &= x(t)	\\
		\end{aligned}	\qquad \quad \right\} \;\; t \in \mathcal{T}	\\
		&\left.
		\begin{aligned}
		\dot{x}  &= f(t,x_t,e,\omega)	
		\end{aligned} \; \right\} \; \mbox{ otherwise,}\ 
		\end{align}
\end{subequations}
are termed the {\em nominal system} $\Sigma_n$, and the impulsive delayed dynamics
\begin{subequations}		\label{eq:Sigmae}
\begin{align}
		&\left.
		\begin{aligned}
		e(t^+) &= h(t,e(t))
		\end{aligned}	\right\} \;\;\; t \in \mathcal{T}	\\
		&\left.
		\begin{aligned}
		\dot{e}  &= g(t,x_t,e_t,\omega_t)
		\end{aligned}   \right\} \; \mbox{ otherwise,} \
		\end{align}
\end{subequations}
are termed the {\em error system} $\Sigma_e$. Observe that $\Sigma_n$ contains delays, but does not depend on $h$ nor ${\cal T}$ as seen from (\ref{eq:Sigman}). Instead, $h$ and ${\cal T}$ constitute the error subsystem $\Sigma_e$ as seen from~(\ref{eq:Sigmae}).

The remainder of our methodology interconnects $\Sigma_n$ and $\Sigma_e$ using appropriate outputs. Basically, $W(i,e)$ from Definition \ref{def:UGES pro} is the output of $\Sigma_e$ while the output of $\Sigma_n$, denoted $H(x_t,\omega_t)$, is obtained from $g(t,x_t,e_t,\omega_t)$ and $W(i,e)$ as specified in Section \ref{sec:application}. Notice that the outputs $H(x_t,\omega_t)$ and $W(i,e)$ are auxiliary signals used to interconnect $\Sigma_n$ and $\Sigma_e$ and solve Problem \ref{prob:main problem}, but do not exist physically. Subsequently, the small-gain theorem is employed to infer $\mathcal{L}_p$-stability with bias. Proofs of the upcoming results are in the Appendix.

\subsection{$\mathcal{L}_p$-Stability with Bias of Impulsive Delayed LTI Systems}		\label{sec:Lp stab}

Before invoking the small-gain theorem in the upcoming subsection, let us establish conditions on the transmission interval $\tau$ and delay $d(t)$ that yield $\mathcal{L}_p$-stability with bias for a class of impulsive delayed LTI systems. Clearly, the results of this subsection are later on applied towards achieving $\mathcal{L}_p$-stability with bias and an appropriate $\mathcal{L}_p$-gain of $\Sigma_e$.

Consider the following class of impulsive delayed LTI system
\begin{subequations}	\label{eq:LTI}
	\begin{align}
		\dot{\xi}(t) &= a \xi(t-d(t)) + \tilde{u}(t),	 \qquad	t \notin \mathcal{T}	\label{eq:LTI flow} \\
		\xi(t^+) &= c \xi(t) + \tilde{\nu}(t),	\qquad \qquad \quad	t \in \mathcal{T},		\label{eq:LTI jump}
	\end{align}
\end{subequations}
where $a \in \mathbb{R}$ and $c \in (-1,1)$, initialized with some $\xi_{t_0} \in PC([-\breve{d},0], \mathbb{R})$. In addition, $d(t)$ is a continuous function upper bounded by $\breve{d}$ while $\tilde{u},\tilde{\nu}: \mathbb{R} \rightarrow \mathbb{R}$ denote external inputs and $\tilde{\nu} \in \mathcal{L}_{\infty}$.

\begin{lemma}		\label{lem:UGES}
	Assume $\tilde{u} \equiv 0$, $\tilde{\nu} \equiv 0$ and consider a positive constant $r$. In addition, let $\lambda_1:=\frac{a^2}{r}$, and $\lambda_2:=c^2$ for $c \neq 0$ or merely $\lambda_2 \in (0,1)$ for $c = 0$. If there exist constants $\lambda > 0$, $M > 1$ such that the conditions
	\begin{enumerate}[(I)]
		\item	$\tau \big( \lambda + r + \lambda_1 M e^{-\lambda \tau} \big) < \ln M$, and
		\item $\tau \big( \lambda + r + \frac{\lambda_1}{\lambda_2} e^{\lambda \breve{d}} \big) < -\ln \lambda_2$
	\end{enumerate}
hold, then the system (\ref{eq:LTI}) is UGES and $\|\xi(t)\| \leq \sqrt{M} \|\xi_{t_0}\| e^{-\frac{\lambda}{2} (t-t_0)}$ for all $t \geq t_0$.
\end{lemma}

\noindent
The previous lemma, combined with the work presented in \cite{aanokhin1995}, results in the following theorem.
\begin{theorem}		\label{tm:Lp}
	Suppose that the system given by (\ref{eq:LTI}) is UGES with constants $\lambda > 0$ and $M > 1$ and that $\sup_{t \in \mathbb{R}} \|\tilde{\nu}(t)\| \leq \tilde{K}_{\nu}$. Then, the system (\ref{eq:LTI}) is $\mathcal{L}_p$-stable with bias $\frac{\tilde{K}_{\nu} \sqrt{M}}{e^{\frac{\lambda \varepsilon}{2}} - 1}$ from $\tilde{u}$ to $\xi$ and with gain $\frac{2}{\lambda} \sqrt{M}$ for each $p \in [1,\infty]$.
\end{theorem}

\subsection{Obtaining MATIs via the Small-Gain Theorem}		\label{sec:application}

We are now ready to state and prove the main result of this article. Essentially, we interconnect $\Sigma_n$ and $\Sigma_e$ via suitable outputs (i.e., $H(x_t,\omega_t)$ and $W(i,e)$, respectively), impose the small-gain condition and invoke the small-gain theorem.
\begin{theorem}		\label{thm:main}
		Suppose the underlying UGES protocol, $d_1(t),\ldots,d_l(t)$ and $K_{\nu} \geq 0$ are given. In addition, assume that
			\begin{enumerate}[(a)]
							\item there exists a continuous function $H : PC([-\overline{d},0],\mathbb{R}^{n_x}) \times PC([-d,0],\mathbb{R}^{n_{\omega}}) \rightarrow \mathbb{R}^m$ such that the system $\Sigma_n$ given by (\ref{eq:Sigman}) is $\mathcal{L}_p$-stable from $(W,\omega)$ to $H(x_t,\omega_t)$ for some $p \in [1,\infty]$, i.e., there exist $K_H,\gamma_H \geq 0$ such that
							\begin{align}
								\|H[t,t_0]\|_p \leq& K_H \|x_{t_0}\| + \gamma_H \|(W,\omega)[t,t_0]\|_p,		\label{eq:Lp H}
							\end{align}
							for all $t \geq t_0$, and
				\item there exists $L \geq 0$ and $d:\mathbb{R} \rightarrow \mathbb{R}_+$, $\sup_{t \in \mathbb{R}} d(t) = \breve{d}$, such that for almost all $t \geq t_0$, almost all $e \in \mathbb{R}^{n_e}$ and for all $(i,x_t,\omega_t) \in \mathbb{N}_0 \times PC([-\overline{d},0],\mathbb{R}^{n_x}) \times PC([-d,0],\mathbb{R}^{n_{\omega}})$ it holds that
							\begin{align}
								\Big\langle \frac{\partial W(i,e)}{\partial e}, g(t,x_t,e_t,\omega_t) \Big\rangle \leq& L W(i,e(t-d(t))) + \nonumber \\
																																																			& + \|H(x_t,\omega_t)\|.		\label{eq:complex}
							\end{align}
			\end{enumerate} 
				Then, the NCS (\ref{eq:NCS system}) is $\mathcal{L}_p$-stable with bias from $\omega$ to $(H,e)$ for each $\tau$ for which there exist $M > 1$ and $\lambda > 0$ satisfying (I), (II) and $\frac{2}{\lambda} \sqrt{M} \gamma_H < 1$ with parameters $a=\frac{\overline{a}}{\underline{a}} L$ and $c=\rho$.
\end{theorem}

\begin{remark}	\label{rmk:main}
	According to Problem \ref{prob:main problem}, condition (a) requires the underlying delays to be admissible, i.e., $\{d_1(t),\ldots,d_l(t)\} \in \mathcal{D}$. Condition (a) implies that the nominal system (i.e., the closed-loop system) is robust with respect to intermittent information and disturbances. Besides $\mathcal{L}_p$-stability, typical robustness requirements encountered in the literature include Input-to-State Stability (ISS) and passivity \cite{dtolicIJRNC2014}. Condition (b) relates the current growth rate of $W(i,e)$ with its past values. As shown in Section \ref{sec:example}, all recommendations and suggestions from \cite{dnesic2004} and \cite{wpmhheemels2010} regarding how to obtain a suitable $W(i,e)$ readily apply because $W(i,e)$ characterizes the underlying UGES protocol (and not the plant-controller dynamics).
\end{remark}

\begin{remark} [Zeno-freeness]	\label{rmk:zeno}
	The left-hand sides of conditions (I) and (II) from Lemma \ref{lem:UGES} are nonnegative continuous functions of $\tau \geq 0$ and approach $\infty$ as $\tau \rightarrow \infty$. Also, these left-hand sides equal zero for $\tau=0$. Note that both sides of (I) and (II) are continuous in $\lambda$, $M$, $\lambda_1$, $\lambda_2$ and $\breve{d}$. Hence, for every $\lambda > 0$, $\lambda_1 \geq 0$, $M > 1$, $\lambda_2 \in (0,1)$ and $\breve{d} \geq 0$ there exists $\tau > 0$ such that (I) and (II) are satisfied. Finally, since $\frac{2}{\lambda} \sqrt{M}$ is continuous in $\lambda$ and $M$, we infer that for every finite $\gamma_H > 0$ there exists $\tau > 0$ such that $\frac{2}{\lambda} \sqrt{M} \gamma_H < 1$. In other words, for each admissible $d_i(t)$, $i \in \{1,\ldots,l\}$, the unwanted Zeno behavior is avoided and the proposed methodology does not yield continuous feedback that might be impossible to implement. Notice that each $\tau$ yielding $\frac{2}{\lambda} \sqrt{M} \gamma_H < 1$ is a candidate for~$\overline{\tau}$. Depending on $r$, $\lambda_2$, $\lambda$ and $M$, the maximal such $\tau$ is in fact MATI~$\overline{\tau}$.
\end{remark}

\begin{remark}		\label{rmk:no general 2}
	The right hand side of (\ref{eq:complex}) might not be descriptive enough for many problems of interest. In general, (\ref{eq:complex}) should be sought in the form $\Big\langle \frac{\partial W(i,e)}{\partial e}, g(t,x_t,e_t,\omega_t) \Big\rangle \leq \sum_{k=1}^{m} L_k W(i,e(t-\grave{d}_k(t))) + \|H(x_t,\omega_t)\|$, where $\grave{d}_k: \mathbb{R} \rightarrow \mathbb{R}_+$ and $m \geq 1$. As this general form leads to tedious computations (as evident from the proof of Lemma \ref{lem:UGES} in the Appendix), we postpone its consideration for the future. For the time being, one can intentionally delay the communicated signals in order to achieve a single discrete delay $d(t)$ in (\ref{eq:complex}). This idea is often found in the literature and can be accomplished via the Controller Area Network (CAN) protocol, time-stamping of data and introduction of buffers at receiver ends (refer to \cite{IEEEhowto:jhespanha} and references therein).
\end{remark}

\begin{remark}
Noisy measurements can be a consequence of quantization errors. According to \cite{nmartins2006}, feedback control prone to quantization errors cannot yield closed-loop systems with linear $\mathcal{L}_p$-gains. Hence, the bias term in the linear gain $\mathcal{L}_p$-stability with bias result of Theorem \ref{thm:main} cannot be removed without contradicting the points in \cite{nmartins2006}. Further investigations of quantized feedback are high on our future research agenda.
\end{remark}

\begin{remark}		\label{rmk:dropouts}
Let us consider the case of lossy communication channels. If there is an upper bound on the maximum number of successive dropouts, say $N_d \in \mathbb{N}$, simply use $\frac{\tau}{N_d}$ as the transmission interval in order for Theorem \ref{thm:main} to hold. Moreover, the transmission instants among NCS links need not to be (and often cannot be) synchronized. In this case, each NCS must transmit at a rate smaller than $\tau_{RR}$ (instead of $\tau_{RR} l$), where $\tau_{RR}$ is the MATI obtained for the RR protocol, in order to meet the prespecified performance requirements. Observe that this leads to asynchronous transmission protocols, which in turn increases the likelihood of packet collisions \cite{mmamduhiCDC2014}.
\end{remark}

\begin{corollary}		\label{cor:main}
	Assume that the conditions of Theorem \ref{thm:main} hold and that $x$ is $\mathcal{L}_p$-detectable from $(W,\omega,H)$. Then the NCS (\ref{eq:NCS system}) is $\mathcal{L}_p$-stable with bias from $\omega$ to $(x,e)$.
\end{corollary}

\noindent
In the following proposition, we provide conditions that yield UGS and GAS of the interconnection $\Sigma_n$ and $\Sigma_e$. Recall that $\omega \equiv\textbf{0}_{n_{\omega}}$ and $K_{\nu}=0$ are the disturbance and noise settings, respectively, corresponding to UGS and GAS.

\begin{proposition}		\label{prop:UGAS}
	Assume that the interconnection of systems $\Sigma_n$ and $\Sigma_e$, given by (\ref{eq:Sigman}) and (\ref{eq:Sigmae}), is $\mathcal{L}_p$-stable from $\omega$ to $(x,e)$. If $p=\infty$, then this interconnection is UGS. When $p \in [1,\infty)$, assume that $f(t,x_t,e,\mathbf{0}_{n_{\omega}})$ and $g(t,x_t,e_t,\tilde{\mathbf{0}}_{n_{\omega}})$ are (locally) Lipschitz uniformly in $t$ as well as that $\|H(x_t,\tilde{\mathbf{0}}_{n_{\omega}})\| \rightarrow 0$ as $\|x_t\| \rightarrow 0$. Then, this interconnection is GAS. 
\end{proposition}

\section{Numerical Examples}	\label{sec:example}

\subsection{Constant Delays}

The following example is motivated by \cite[Example 2.2.]{jtsinias1991} and all the results are provided for $p=2$. Consider the following nonlinear delayed plant (compare with (\ref{eq:plant}))
\begin{align*}
	&\begin{bmatrix}
			\dot{x}_{p1}(t)	\\
			\dot{x}_{p2}(t)
	\end{bmatrix} =	\nonumber \\
	&\begin{bmatrix}
										\begin{smallmatrix}
										-0.5 x_{p1}(t) + x_{p2}(t) - 0.25x_{p1}(t) \sin \big(u(t) x_{p2}(t-d_{p1}) \big)	\\
										x_{p1}(t) \sin \big(u(t) x_{p2}(t-d_{p1}) \big) + 1.7x_{p2}(t-d_{p2}) + u(t) - x_{p2}(t)
										\end{smallmatrix}
									\end{bmatrix} +	\nonumber \\
								& \qquad  + \begin{bmatrix}
																		\begin{smallmatrix}
																			\omega_1(t)	\\
																			\omega_2(t)
																		\end{smallmatrix}
																	\end{bmatrix}	
\end{align*}
controlled with (compare with (\ref{eq:controller}))
\begin{align*}
	u(t) = -2 x_{p1}(t) - 2 x_{p2}(t).
\end{align*}
As this controller is without internal dynamics, therefore $x(t):=x_p(t)=(x_{p1}(t),x_{p2}(t))$. Additionally, $\omega(t):=(\omega_1(t),\omega_2(t))$.

Let us consider the NCS setting in which noisy information regarding $x_{p1}$ and $x_{p2}$ are transmitted over a communication network while the control signal is not transmitted over a communication network nor distorted (i.e., $\hat{u}=u$). In addition, consider that the information regarding $x_{p2}$ arrives at the controller with delay $d$ while information regarding $x_{p1}$ arrives in timely manner. For the sake of simplicity, let us take $d=d_{p2}$. Apparently, the output of the plant is $y(t)=x_p(t)=x(t)$ and there are two NCS links so that $l=2$. Namely, $x_{p1}$ is transmitted through one NCS link while $x_{p2}$ is transmitted through the second NCS link. The repercussions of these two NCS links are modeled via the following error vector (compare with (\ref{sensing error}))
\begin{align*}
	e = \begin{bmatrix}
					e_1		\\
					e_2
			\end{bmatrix} = \hat{y} - \Bigg( \underbrace{\begin{bmatrix}
																												x_{p1}(t)	\\
																												0
																										\end{bmatrix} + \begin{bmatrix}
																																				0	\\
																																				x_{p2}(t-d)
																																		\end{bmatrix}}_{y_t} \Bigg).
\end{align*}
The expressions (\ref{eq:nesic 1}) and (\ref{eq:nesic 2}) for this example become:
\begin{align}
	&\dot{x}(t) = \underbrace{\begin{bmatrix}
									-0.5	&		1	\\
									-2	  &	 -1
								\end{bmatrix}}_{A_1} x(t) + \underbrace{\begin{bmatrix}
																														0		&			 0		\\
																														0		&		-0.3
																												\end{bmatrix}}_{A_2} x(t-d) +		
\end{align}
\begin{align}
							&+\underbrace{\begin{bmatrix}
																																																			-0.25		&		0		\\
																																																					1		&		0
																																																	\end{bmatrix}}_{B_1} x(t) N(x_t,e)
							+ \underbrace{\begin{bmatrix}
											0		&		0		\\
									 	 -2		&	 -2		
								\end{bmatrix}}_{B} e(t) + \omega(t),			\label{eq:nominal ex}
\end{align}	
\begin{align}							
	\dot{e}(t) = \dot{\hat{y}} &- B e(t-d) + \underbrace{\begin{bmatrix}
																								0.5		&		-1		\\
																									0		&		 0
																						\end{bmatrix}}_{C_1} x(t) + \nonumber \\
							&+\underbrace{\begin{bmatrix}
																																												0		&		0	\\
																																												2		&		2	
																																					    			 \end{bmatrix}}_{C_2} x(t-d) + \underbrace{\begin{bmatrix}
																																																																	0		&	 	0		\\
																																																																	0	  &	0.3		
																																																															 \end{bmatrix}}_{C_3} x(t-2d) +	
\end{align}
\begin{align}
							&+ \underbrace{\begin{bmatrix}
																0.25		&		0	\\
																	 0	 	&		0
															\end{bmatrix}}_{C_4} x(t) N(x_t,e) + \underbrace{\begin{bmatrix}
																																										   0		&		0	\\
																																											-1	 	&		0
																																								\end{bmatrix}}_{C_5} x(t-d) N(x_t,e_t) + \nonumber \\
							&+ \underbrace{\begin{bmatrix}
																																																																				-1		&		0		\\
																																																																					0		&		0
																																																																			\end{bmatrix}}_{C_6} \omega(t) +
										\underbrace{\begin{bmatrix}
															  0		&		0		\\
																0		&	 -1
										\end{bmatrix}}_{C_7} \omega(t-d),		\nonumber
\end{align}
where $N(x_t,e):= \sin \Bigg( \Big[ -2 \big(x_{p1}(t)+e_1(t) \big) - 2\big( x_{p2}(t-d) + e_2(t) \big) \Big]  x_{p2}(t-d_{p1}) \Bigg)$ and $N(x_t,e_t):= \sin \Bigg( \Big[ -2 \big(x_{p1}(t-d)+e_1(t-d) \big) - 2\big( x_{p2}(t-2d) + e_2(t-d) \big) \Big]  x_{p2}(t-d_{p1}-d) \Bigg)$.

According to \cite{dnesic2004} and \cite{wpmhheemels2010}, we select $W_{RR}(i,e):=\|D(i) e\|$ and $W_{TOD}(t,e):=\|e\|$, where $D(i)$ is a diagonal matrix whose diagonal elements are lower bounded by $1$ and upper bounded by $\sqrt{l}$. Next, we determine $L_{RR}$, $H_{RR}(x,\omega,d)$, $L_{TOD}$ and $H_{TOD}(x,\omega,d)$ from Theorem \ref{thm:main} for the ZOH strategy (i.e., $\dot{\hat{y}} \equiv \textbf{0}_{n_y}$) obtaining (\ref{eq:long 1}) and (\ref{eq:long 2}).
\begin{figure*}
\begin{align}
	&\Big\langle \frac{\partial W_{RR}(i,e)}{\partial e}, \dot{e} \Big\rangle \leq \|D(i) \dot{e}\| \leq  \underbrace{ \sqrt{l} \|B\|}_{L_{RR}} \underbrace{\|D(i)e(t-d)\|}_{W_{RR}(i,e(t-d))} +     \nonumber \\
					&+ \underbrace{\sqrt{l} \big( \|C_1 x(t) + C_2 x(t-d) + C_3 x(t-2d) + C_6 \omega(t) + C_7 \omega(t-d)\| + \|C_4 x(t)\| + \|C_5 x(t-d)\| \big)}_{H_{RR}(x_t,\omega_t)},	\label{eq:long 1} \\
	&\Big\langle \frac{\partial W_{TOD}(i,e)}{\partial e}, \dot{e} \Big\rangle \leq \underbrace{\|B\|}_{L_{TOD}} \underbrace{\|e(t-d)\|}_{W_{TOD}(i,e(t-d))}	+ \nonumber \\
					&+ \underbrace{ \big( \|C_1 x(t) + C_2 x(t-d) + C_3 x(t-2d) + C_6 \omega(t) + C_7 \omega(t-d)\| + \|C_4 x(t)\| + \|C_5 x(t-d)\| \big)}_{H_{TOD}(x_t,\omega_t)}, \label{eq:long 2}
\end{align}
\end{figure*}
In order to estimate $\gamma_H$, we utilize Lyapunov-Krasovskii functionals according to \cite[Chapter 6]{sboyd1994} and \cite{dfcoutinho2008}. Basically, if there exist $\gamma \geq 0$ and a Lyapunov-Krasovskii functional $V(x_t)$ for the nominal system (\ref{eq:Sigman}), that is (\ref{eq:nominal ex}), with the input $(W,\omega)$ and the output $H$ such that its time-derivative along the solution of (\ref{eq:Sigman}) with a zero initial condition satisfies:
\begin{align}
	\dot{V}(x_t) + H^{\top} H - \gamma^2 (W,&\omega)^{\top}(W,\omega) \leq 0, \nonumber \\
																					&\forall x_t \in C([-\overline{d},0],\mathbb{R}^{n_x}),		\label{eq:neg def}
\end{align}
than the corresponding $\mathcal{L}_2$-gain $\gamma_H$ is less than $\gamma$.
The functional used herein is
\begin{align}
	V(x_t) = x(t)^{\top} C x(t) + \int_{-d}^{0} x(t+\theta)^{\top} E x(t+\theta) \textrm{d}\theta,		\label{eq:LK functional}
\end{align}
where $C$ and $E$ are positive-definite symmetric matrices. 

Next, we illustrate the steps behind employing (\ref{eq:LK functional}). Let us focus on TOD (i.e., the input is $(e,\omega)$) and the output $C_1 x(t) + C_2 x(t-d) + C_6 \omega(t)$. The same procedure is repeated for the remaining terms of $H_{RR}(x_t,\omega_t)$ and $H_{TOD}(x_t,\omega_t)$. For the Lyapunov-Krasovskii functional in (\ref{eq:LK functional}), the expression (\ref{eq:neg def}) boils down to the Linear Matrix Inequality (LMI) (see \cite{sboyd1994} for more) given by (\ref{eq:LMI}).
\begin{figure*}
\begin{align}
	\begin{bmatrix}
		\begin{smallmatrix}
			A_1^{\top} C + C A_1 + E + N(x_t,e) (B_1^{\top} C + C B_1) + C_1^{\top}C_1 	&	C A_2 + C_1^{\top} C_2			&	C B 											&		C+C_1^{\top} C_6							\\
			A_2^{\top} C + C_2^{\top} C_1																								&		-E+C_2^{\top} C_2					&	\textbf{0}_{2 \times 2}		&		C_2^{\top} C_6								\\
			B^{\top} C																																	&		\textbf{0}_{2 \times 2}		&	-\gamma^2 I_2							&		\textbf{0}_{2 \times 2}				\\
			C+C_6^{\top} C_1																														&		C_6^{\top} C_2						&	\textbf{0}_{2 \times 2}		&		-\gamma^2 I_2 + C_6^{\top} C_6
		\end{smallmatrix}
	\end{bmatrix} \leq 0.				\label{eq:LMI}	\\
	\midrule	\nonumber
\end{align}
\end{figure*}
Notice that the above LMI has to hold for all $N(x_t,e) \in [-1,1]$. Using the LMI Toolbox in MATLAB, we find that the minimal $\gamma$ for which (\ref{eq:LMI}) holds is in fact $\gamma_H$. For our TOD example and the specified output, we obtain $\gamma_H=18.7051$. This $\gamma_H$ holds for all $d \geq 0$. In other words, any $d \geq 0$ is an admissible delay and belongs to the family $\mathcal{D}$. For RR, simply multiply $\gamma_H$ by $\sqrt{2}$.

Detectability of $x$ from $(W,x,H)$, which is a condition of Corollary \ref{cor:main}, is easily inferred by taking $x(t)$ to be the output of the nominal system and computing the respective $\mathcal{L}_2$-gain $\gamma_d$. Next, let us take the output of interest to be $x$ and find MATIs that yield the desired $\mathcal{L}_p$-gain from $\omega$ to $x$ to be $\gamma_{\mathrm{des}}=50$. Combining (\ref{eq:Lp e}) with $\gamma_d$ leads to the following condition
\begin{align*}
		\gamma_W \gamma_H < 1 - \frac{\gamma_d}{\gamma_{\mathrm{des}}}
\end{align*}
that needs to be satisfied (by changing $\gamma_W$ through changing MATIs) in order to achieve the desired gain $\gamma_{\mathrm{des}}$. In addition, observe that the conditions of Proposition \ref{prop:UGAS} hold (and the closed-loop system is an autonomous system) so that we can infer UGAS when $\omega \equiv\textbf{0}_{n_{\omega}}$ and $K_{\nu}=0$.

Let us now introduce the following estimator (compare with (\ref{eq:no estimation}))
\begin{align}
	\dot{\hat{y}} = B \hat{y}(t-d) = B \Bigg( e(t&-d) + \begin{bmatrix}
																													1		&		0	\\
																													0		&		0
																										\end{bmatrix} x(t-d) +	\nonumber \\
																									 &+		\begin{bmatrix}
																																								0		&		0	\\
																																								0		&		1
																																						\end{bmatrix} x(t-2d)	\Bigg),			\label{eq:estimator}
\end{align}
which can be employed when one is interested in any of the three performance objectives (i.e., UGAS, $\mathcal{L}_p$-stability or $\mathcal{L}_p$-stability with a desired gain).

\begin{figure*}
  \centering
  \subfigure[] {\label{fig:RR}\includegraphics[width=0.49\textwidth]{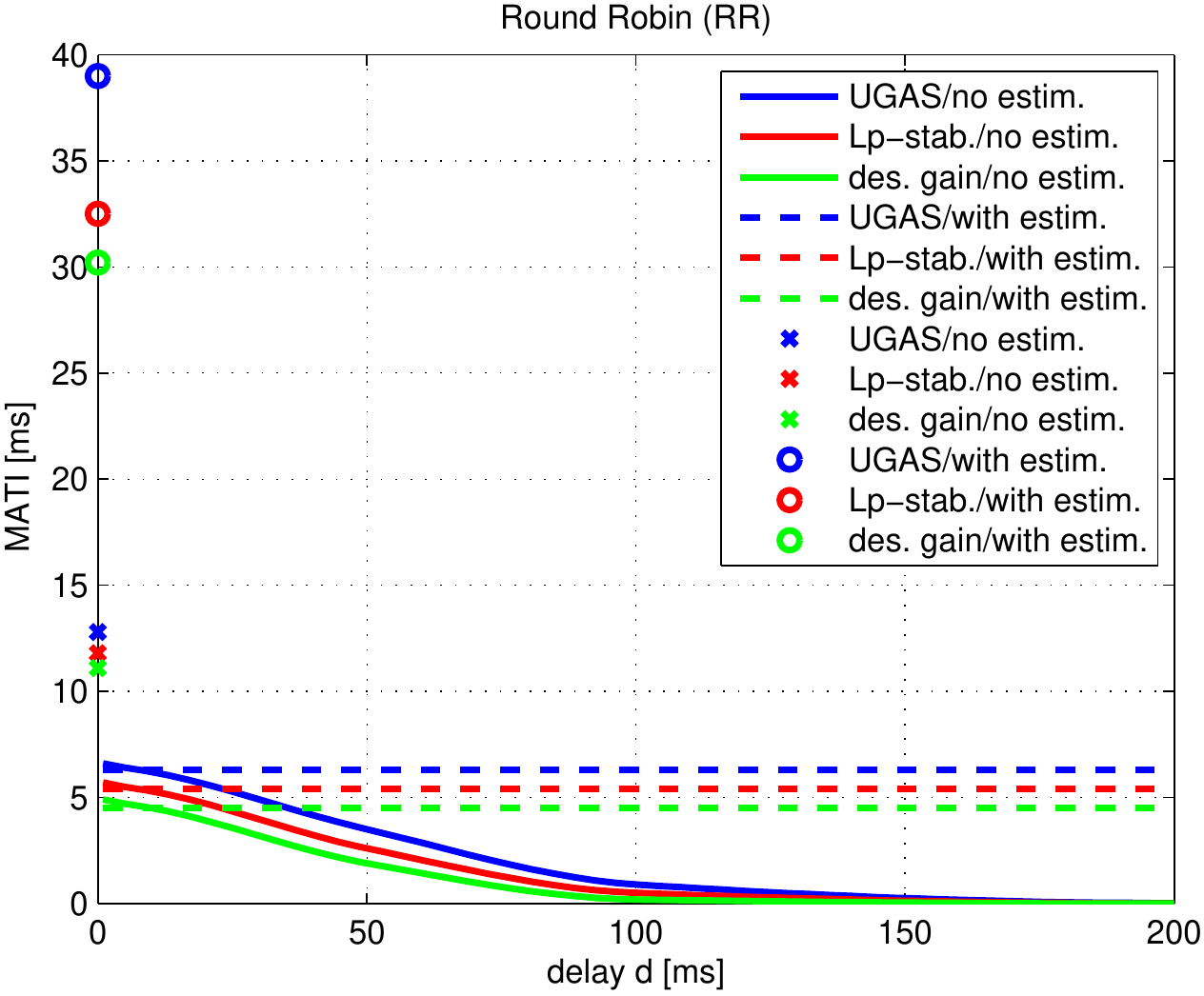}}
  \subfigure[] {\label{fig:TOD}\includegraphics[width=0.49\textwidth]{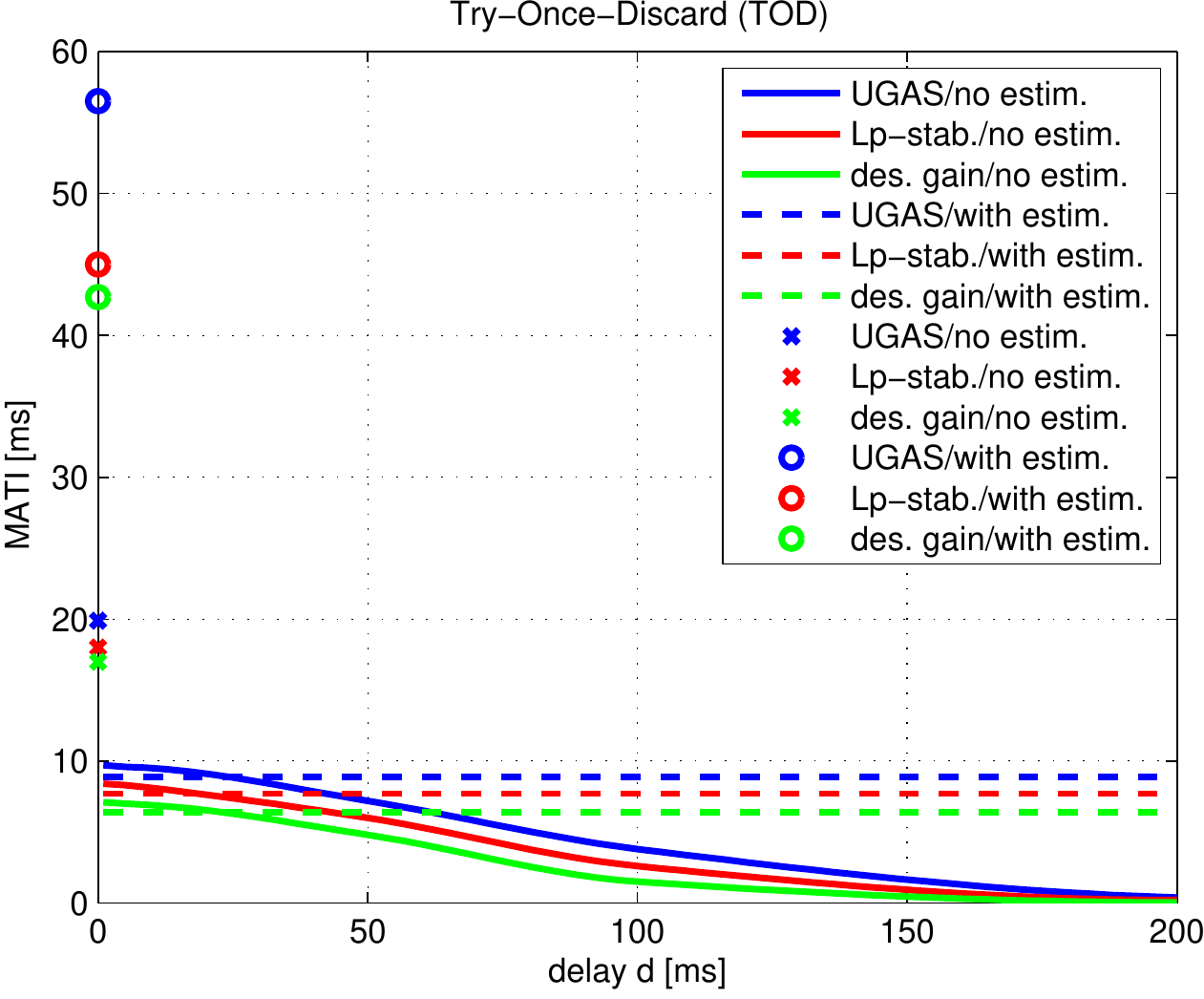}}
  \caption{Numerically obtained MATIs for different delay values $d \geq 0$ in scenarios with and without estimation:
					\subref{fig:RR} RR; and,
					\subref{fig:TOD} TOD.}
  \label{fig:protocols}
\end{figure*}

Figure \ref{fig:protocols} provides evidence that the TOD protocol results in greater MATIs (at the expense of additional implementation complexity/costs) and that the model-based estimators significantly prolong MATIs, when compared with the ZOH strategy, especially as $d$ increases. We point out that a different estimator (such as $\dot{\hat{y}}=k B\hat{y}(t-d)$ for some $k \in \mathbb{R}$) can be employed as $d$ approaches zero (because the estimator slightly decreases the MATIs as seen in Figure \ref{fig:protocols}) to render greater MATIs in comparison with the scenarios without estimation. In addition, notice that the case $d=0$ boils down to ODE modeling so that we can employ less conservative tools for computing $\mathcal{L}_2$-gains. Accordingly, the $4 \times 4$ LMI given by (\ref{eq:LMI}) becomes a $3 \times 3$ LMI resulting in a smaller $\gamma_H$. Furthermore, the constant $a$ in Theorem \ref{thm:main} becomes $L$, rather than $\frac{\overline{a}}{\underline{a}} L$, which in turn decreases $\gamma_W$ for the same $\tau$. Apparently, MATIs pertaining to UGAS are greater than the MATIs pertaining to $\mathcal{L}_p$-stability from $\omega$ to $(x,e)$ and these are greater than the MATIs pertaining to $\mathcal{L}_p$-stability from $\omega$ to $x$ with $\gamma_{\mathrm{des}}=50$. 

For completeness, we provide the gains used to obtain Figure \ref{fig:protocols}: $\gamma_{H,TOD}=9.6598$ for UGAS with ZOH and $d=0$; $\gamma_{H,TOD}=4.3344$ for UGAS with estimation and $d=0$; $\gamma_{H,TOD}=22.3631$ for UGAS with ZOH and $d > 0$; $\gamma_{H,TOD}=27.3659$ for UGAS with estimation and $d > 0$; $\gamma_{H,TOD}=10.8958$ for $\mathcal{L}_p$-stability with ZOH and $d=0$; $\gamma_{H,TOD}=5.3258$ for $\mathcal{L}_p$-stability with estimation and $d = 0$; $\gamma_{H,TOD}=26.4601$ for $\mathcal{L}_p$-stability with ZOH and $d > 0$; $\gamma_{H,TOD}=31.7892$ for $\mathcal{L}_p$-stability with estimation and $d > 0$; $\gamma_d=3.5884$ for $d = 0$; and, $\gamma_d=7.9597$ for $d > 0$. Recall that $\gamma_{H,RR}=\sqrt{2} \gamma_{H,TOD}$.

\subsection{Time-Varying Delays}

The following example is taken from \cite{ptallapragada2011,dtolicMED12} and the results are provided for $p=2$. Consider the inverted pendulum (compare with (\ref{eq:plant})) given by
\begin{align*}
		\dot{x}_{p1} &= x_{p2} + \omega_1		\\
		\dot{x}_{p2} &= \frac{1}{L} (-g \cos(x_{p1}) + u ) + \omega_2,
\end{align*}
where $g = 9.8$ and $L=2$, controlled with 
\begin{align*}
		u=-L\lambda x_{p2} + g \cos(x_{p1}) - K(x_{p2}+\lambda x_{p1}),
\end{align*}
where $K=50$ and $\lambda=1$. Clearly, the control system goal is to keep the pendulum at rest in the upright position. As this controller is without internal dynamics, therefore $x(t):=x_p(t)=(x_{p1}(t),x_{p2}(t))$. Additionally, $\omega(t):=(\omega_1(t),\omega_2(t))$.

Consider the NCS setting in which noisy information regarding $x_{p1}$ and $x_{p2}$ are transmitted over a communication network while the control signal is not transmitted over a communication network nor distorted (i.e., $\hat{u}=u$). In addition, consider that the information regarding $x_{p2}$ arrives at the controller with delay $d(t) \leq \breve{d}$ and $|\dot{d}(t)| \leq \breve{d}_1$ while information regarding $x_{p1}$ arrives instantaneously. Apparently, the output of the plant is $y(t)=x_p(t)=x(t)$ and there are two NCS links so that $l=2$. Namely, $x_{p1}$ is transmitted through one NCS link while $x_{p2}$ is transmitted through the second NCS link. The repercussions of these two NCS links are modeled via the following error vector (compare with (\ref{sensing error}))
\begin{align*}
	e = \begin{bmatrix}
					e_1		\\
					e_2
			\end{bmatrix} = \hat{y} - \Bigg( \underbrace{\begin{bmatrix}
																												x_{p1}(t)	\\
																												0
																										\end{bmatrix} + \begin{bmatrix}
																																				0	\\
																																				x_{p2}(t-d(t))
																																		\end{bmatrix}}_{y_t} \Bigg).
\end{align*}
The expressions (\ref{eq:nesic 1}) and (\ref{eq:nesic 2}) for this example become:
\begin{align}
	&\dot{x}(t) = \underbrace{\begin{bmatrix}
											0 	&		1	\\
		\frac{-K \lambda}{L}  &	 -1
								\end{bmatrix}}_{A_1} x(t) + \underbrace{\begin{bmatrix}
																														0		&													 0		\\
																														0		&		\frac{-K}{L} - \lambda L 
																												\end{bmatrix}}_{A_2} x(t-d(t)) +	\nonumber	
\end{align}
\begin{align}
							&+\begin{bmatrix}
											0		\\
											n(x_1(t),e_1(t))
							\end{bmatrix}
							+ \underbrace{\begin{bmatrix}
											0		&		0												\\
	\frac{-K \lambda}{L}		&	 \frac{-K}{L} - \lambda L		
								\end{bmatrix}}_{B} e(t) + \omega(t),	\nonumber
\end{align}	
\begin{align}							
	&\dot{e}(t) = \dot{\hat{y}} + \underbrace{\begin{bmatrix}
																								  0		&		-1		\\
																									0		&		 0
																						\end{bmatrix}}_{B_1} x(t) + \underbrace{\begin{bmatrix}
																																											\begin{smallmatrix}
																													-1	&		0		\\
															 											  		0		&		0
																								\end{smallmatrix}
																								\end{bmatrix}}_{C_1} \omega(t) + \qquad \quad \qquad \qquad  \nonumber 
\end{align}
\begin{align}
							&+ \Bigg( -B e(t-d(t)) + \underbrace{\begin{bmatrix}
																																	0		&		0	\\
																								\frac{K \lambda}{L}		&		0	
																				\end{bmatrix}}_{B_2} x(t-d(t)) - 	\nonumber
\end{align}
\begin{align}
							&- \begin{bmatrix}
																							0	\\
											n(x_1(t-d(t)),e_1(t-d(t)))
								  \end{bmatrix} +  \underbrace{\begin{bmatrix}
																										\begin{smallmatrix}
																												0		&		0		\\
																												0		&	 -1
																										\end{smallmatrix}
																							\end{bmatrix}}_{C_2} \omega(t-d(t)) + \nonumber \\
							&+ \underbrace{\begin{bmatrix}
																	0		&	 											0		\\
																	0	  & \frac{K}{L} + \lambda L		
															\end{bmatrix}}_{B_3} x(t-2d(t)) \Bigg) (1-\dot{d}(t)),		\nonumber
\end{align}
where $n(x_1(t),e_1(t)) = \frac{-2g}{L} \sin\Big( \frac{e_1(t)+2x_1(t)}{2} \Big) \sin\Big( \frac{e_1(t)}{2} \Big)$.

According to \cite{dnesic2004} and \cite{wpmhheemels2010}, we select $W_{RR}(i,e):=\|D(i) e\|$ and $W_{TOD}(t,e):=\|e\|$, where $D(i)$ is a diagonal matrix whose diagonal elements are lower bounded by $1$ and upper bounded by $\sqrt{l}$. Next, we determine $L_{RR}$, $H_{RR}(x,\omega,d)$, $L_{TOD}$ and $H_{TOD}(x,\omega,d)$ from Theorem \ref{thm:main} for the ZOH strategy (i.e., $\dot{\hat{y}} \equiv \textbf{0}_{n_y}$) obtaining (\ref{eq:long 3}) and (\ref{eq:long 4}).
\begin{figure*}
\begin{align}
	\Big\langle \frac{\partial W_{RR}(i,e)}{\partial e},& \dot{e} \Big\rangle \leq \|D(i) \dot{e}\| \leq  \underbrace{ \sqrt{l} \big(1+\breve{d}_1\big) \Big( \|B\| + \frac{g}{L} \Big)}_{L_{RR}} \underbrace{\|D(i)e(t-d(t))\|}_{W_{RR}(i,e(t-d(t)))} +     \nonumber \\
					&+ \underbrace{\sqrt{l} \Big( \big\|B_1 x(t) + C_1 \omega(t) + \big(1-\dot{d}(t)\big) \big(B_2 x(t-d(t)) + B_3 x(t-2d(t))  + C_2 \omega(t-d(t)) \big) \big\| \Big)}_{H_{RR}(x_t,\omega_t)},	\label{eq:long 3} \\
	\Big\langle \frac{\partial W_{TOD}(i,e)}{\partial e},& \dot{e} \Big\rangle \leq \underbrace{ \big(1+\breve{d}_1 \big) \Big(\|B\| + \frac{g}{L} \Big)}_{L_{TOD}} \underbrace{\|e(t-d(t))\|}_{W_{TOD}(i,e(t-d(t)))}	+ \nonumber \\
					&+ \underbrace{ \Big( \big\|B_1 x(t) + C_1 \omega(t) + \big(1-\dot{d}(t) \big) \big( B_2 x(t-d(t)) + B_3 x(t-2d(t)) + C_2 \omega(t-d(t)) \big) \big\| \Big)}_{H_{TOD}(x_t,\omega_t)}, \label{eq:long 4}	\\ 
					\midrule \nonumber
\end{align}
\end{figure*}

The Lyapunov-Krasovskii functional used for the pendulum example is
\begin{align*}
	&V(t,x_t,\dot{x}_t) = x(t)^{\top} P x(t) + \int_{t-\breve{d}}^{t} x(s)^{\top} S x(s) \textrm{d}s +	\nonumber \\
			&+ \breve{d} \int_{-\breve{d}}^{0} \int_{t+\theta}^{t}  \dot{x}(s)^{\top} R \dot{x}(s) \textrm{d}s \textrm{d}\theta + \int_{t-d(t)}^{t} x(s)^{\top} Q x(s) \textrm{d}s,
\end{align*}
where $P$ is a positive-definite symmetric matrix while $S$, $R$ and $Q$ are positive-semidefinite symmetric matrices. Next, we take the output of interest to be $x$ and find MATIs that yield the desired $\mathcal{L}_p$-gain from $\omega$ to $x$ to be $\gamma_{\mathrm{des}}=15$. In addition, observe that the conditions of Proposition \ref{prop:UGAS} hold (and the closed-loop system is an autonomous system) so that we can infer UGAS when $\omega \equiv\textbf{0}_{n_{\omega}}$ and $K_{\nu}=0$.

We use the following estimator (compare with (\ref{eq:no estimation}))
\begin{align*}
	\dot{\hat{y}} &= B \hat{y}(t-d(t)) (1-\dot{d}(t)) = B \Bigg( e(t-d(t)) +	\nonumber \\
								&+ \begin{bmatrix}
																													1		&		0	\\
																													0		&		0
																										\end{bmatrix} x(t-d(t)) +		\begin{bmatrix}
																																											0		&		0	\\
																																											0		&		1
																																								\end{bmatrix} x(t-2d(t)) \Bigg)  (1-\dot{d}(t)),
\end{align*}
which can be employed in any of the three performance objectives (i.e., UGAS, $\mathcal{L}_p$-stability or $\mathcal{L}_p$-stability with a desired gain) provided $d(t)$ is known. One can use the ideas from Remark \ref{rmk:no general 2} towards obtaining known delays.

\begin{figure*}
  \centering
  \subfigure[] {\label{fig:RR2}\includegraphics[width=0.49\textwidth]{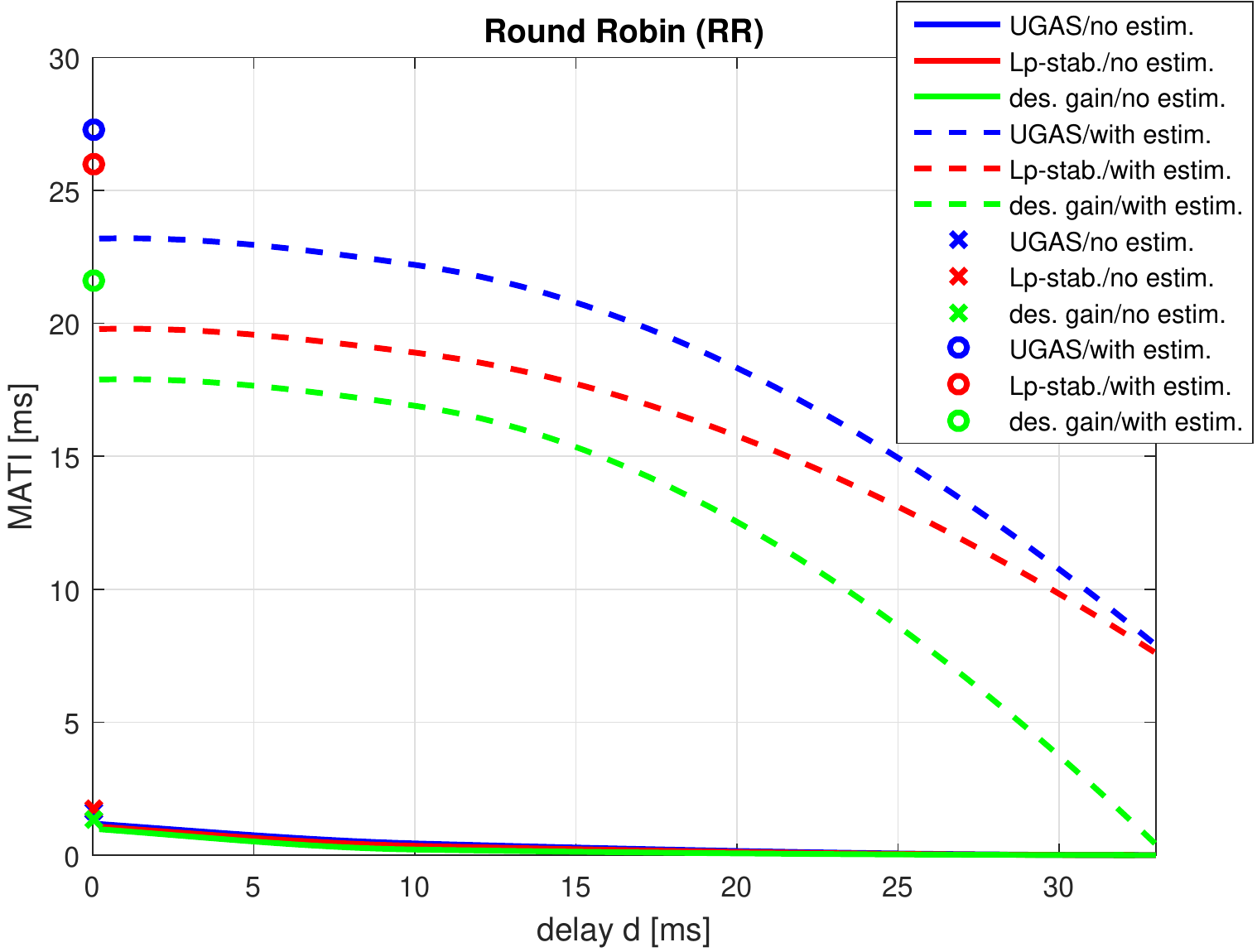}}
  \subfigure[] {\label{fig:TOD2}\includegraphics[width=0.49\textwidth]{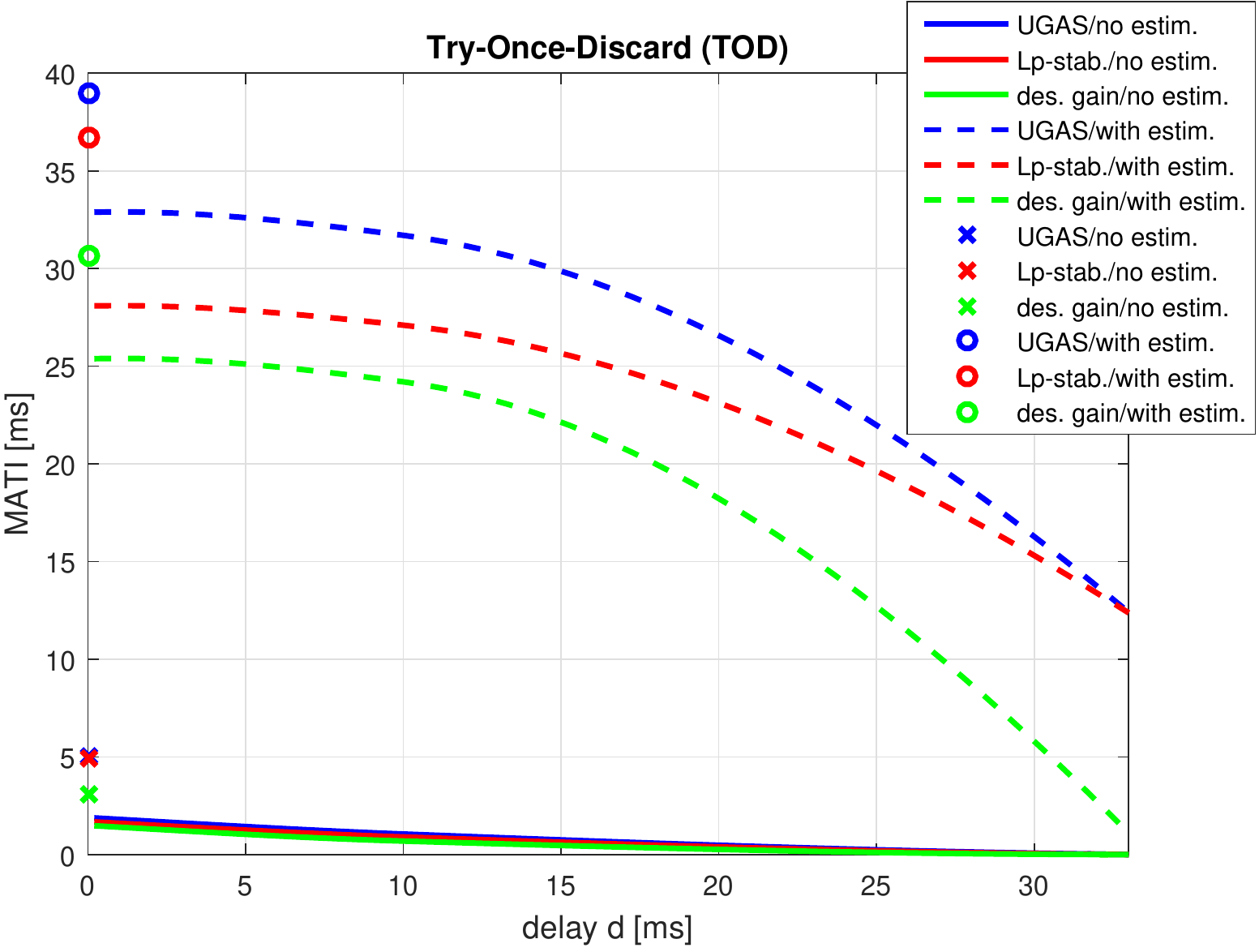}}
  \caption{Numerically obtained MATIs for various constant delay values $d \geq 0$ in scenarios with and without estimation:
					\subref{fig:RR2} RR; and,
					\subref{fig:TOD2} TOD.}
  \label{fig:protocols2}
\end{figure*}

\begin{figure*}
  \centering
  \subfigure[] {\label{fig:RR3}\includegraphics[width=0.49\textwidth]{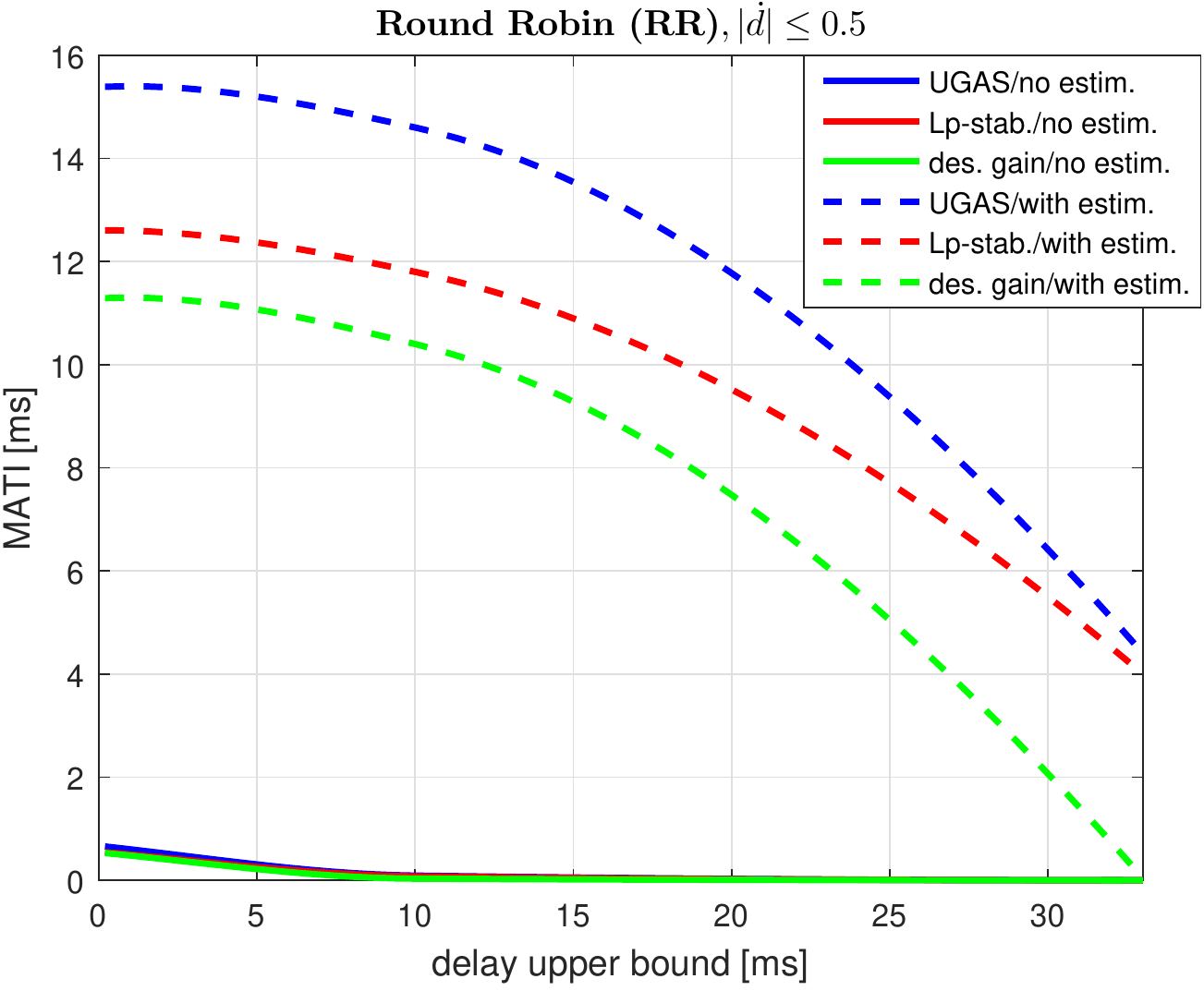}}
  \subfigure[] {\label{fig:TOD3}\includegraphics[width=0.49\textwidth]{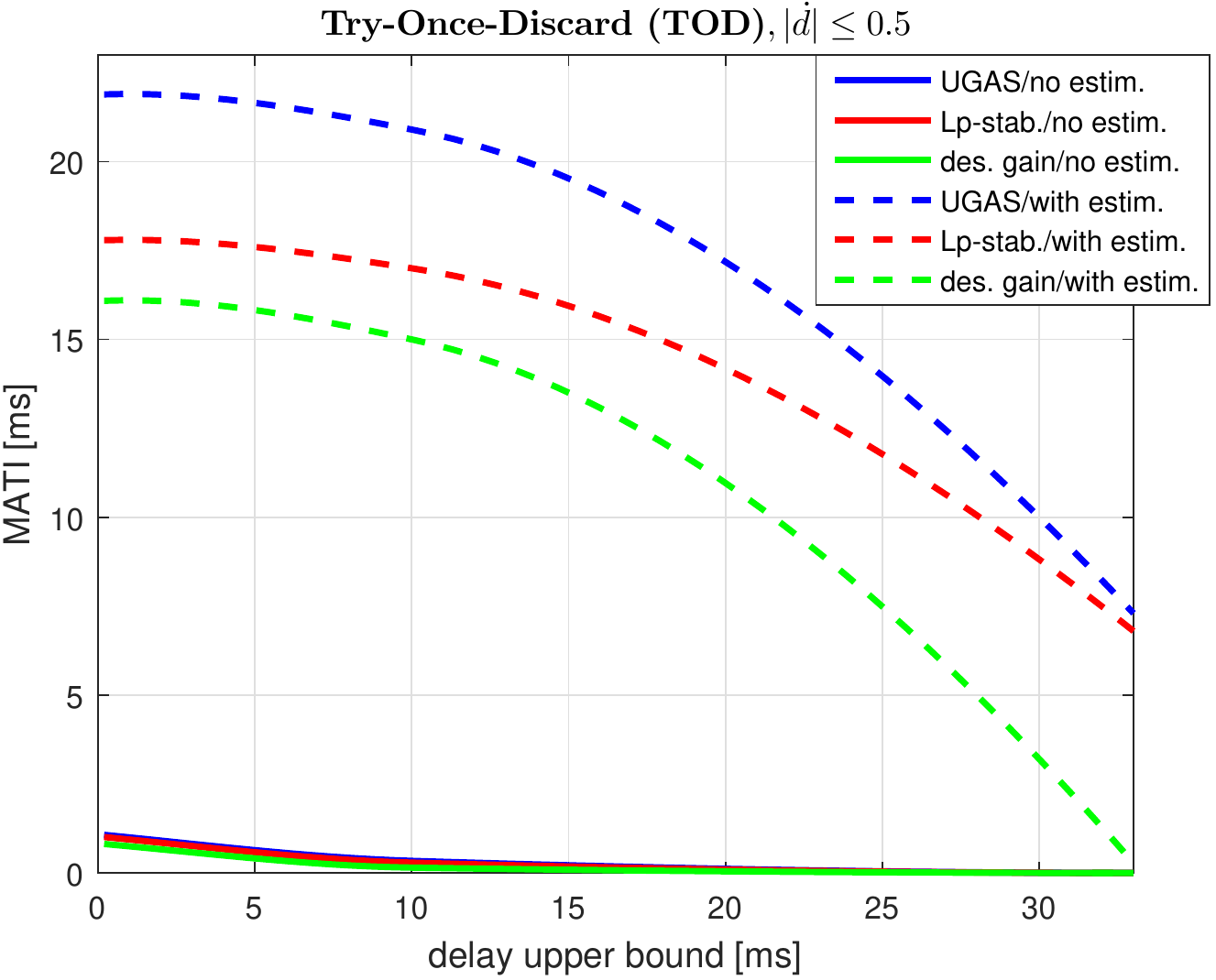}}
  \caption{Numerically obtained MATIs for various time-varying delays $d(t)$ such that $d(t) \leq \breve{d}$ and $|\dot{d}(t)| \leq \breve{d}_1 = 0.5$ in scenarios with and without estimation:
					\subref{fig:RR3} RR; and,
					\subref{fig:TOD3} TOD.}
  \label{fig:protocols3}
\end{figure*}

Figures \ref{fig:protocols2} and \ref{fig:protocols3} provide evidence that the TOD protocol results in greater MATIs (at the expense of additional implementation complexity/costs) and that the model-based estimators significantly prolong MATIs, when compared with the ZOH strategy. In addition, notice that the case $\breve{d}=0$ boils down to ODE modeling so that we can employ less conservative tools for computing $\mathcal{L}_2$-gains. Apparently, MATIs pertaining to UGAS are greater than the MATIs pertaining to $\mathcal{L}_p$-stability from $\omega$ to $(x,e)$ and these are greater than the MATIs pertaining to $\mathcal{L}_p$-stability from $\omega$ to $x$ with $\gamma_{\mathrm{des}}=15$. As expected, time-varying delays upper bounded with some $\breve{d}$ lead to smaller MATIs when compared to constant delays $\breve{d}$. It is worth mentioning that $\breve{d}=33$ ms is the maximal value for which we are able to establish condition (a) of Theorem \ref{thm:main}. Consequently, the delays from Figures \ref{fig:protocols2} and \ref{fig:protocols3} are instances of admissible delays. The exhaustive search for admissible delays is an open problem that is out of scope of this article.

\section{Conclusion}		\label{sec:concl}

In this article, we study how much information exchange between a plant and controller can become intermittent (in terms of MATIs) such that the performance objectives of interest are not compromised. Depending on the noise and disturbance setting, the performance objective can be UGAS or $\mathcal{L}_p$-stability (with a prespecified gain and towards the output of interest). Our framework incorporates time-varying delays and transmission intervals that can be smaller than the delays, plants/controllers with delayed dynamics, external disturbances (or modeling uncertainties), UGES scheduling protocols (e.g., RR and TOD protocols), distorted data and model-based estimators. As expected, the TOD protocol results in greater MATIs than the RR protocol. Likewise, estimation (rather than the ZOH strategy) in between two consecutive transmission instants extends the MATIs.

The primary goal of our future work is to devise conditions rendering $\mathcal{L}_p$-stability of the error dynamics involving several time-varying delays (see Remark \ref{rmk:no general 2}). In addition, in light of \cite{vsdolk2014}, we plan to design event- and self-triggered realizations of our approach.


\section*{Appendix}

\subsection{Proof of Lemma \ref{lem:UGES}}		\label{sec:UGES proof}

This proof follows the exposition in \cite{jzhou2009}. The following two definitions regarding (\ref{eq:general hybrid sys}) are utilized in this proof and are taken from \cite{jzhou2009}.

\begin{definition}	[Lyapunov Function]		\label{def:lyapunov}
The function $V : [t_0,\infty) \times \mathbb{R}^{n_{\xi}} \rightarrow \mathbb{R}^+$ is said to belong to the class $\nu_0$ if we have the following:
\begin{enumerate}
\item $V$ is continuous in each of the sets $[t_{k-1}, t_k) \times \mathbb{R}^{n_{\xi}}$, and for each $\xi \in \mathbb{R}^{n_{\xi}}$ and each $t \in [t_{k-1}, t_k)$, where $k \in \mathbb{N}$, the limit $\lim_{(t,y) \rightarrow (t^-_k,x)} V(t, y) = V (t^-_k, x)$ exists;
\item $V(t, \xi)$ is locally Lipschitz in all $\xi \in \mathbb{R}^{n_{\xi}}$; and 
\item $V (t, 0) \equiv 0$ for all $t \geq t_0$.
\end{enumerate}
\end{definition}

\begin{definition}	[Upper Dini Derivative]		\label{def:dini}
Given a function $V : [t_0,\infty) \times \mathbb{R}^{n_{\xi}} \rightarrow \mathbb{R}^+$, the upper right-hand derivative of $V$ with respect to system (\ref{eq:general hybrid sys}) is
defined by $D^+ V (t, \xi(t)) = \limsup_{\delta \searrow 0} \frac{1}{\delta} [V (t+\delta, \xi(t+\delta)) - V (t, \xi(t))]$.
\end{definition}

\begin{proof}
	We prove this theorem employing mathematical induction. Consider the following Lyapunov function for (\ref{eq:LTI}) with $\tilde{u} \equiv 0$, $\tilde{\nu} \equiv 0$:
	\begin{align}
		V(t,\xi(t)) = r \xi(t)^2.			\label{eq:lyap bounds}
	\end{align}
	Using $2 a b \leq a^2 + b^2$, $a,b \in \mathbb{R}$ in what follows, we obtain
	\begin{align}
		D^+ V(t,&\xi(t)) \leq 2 r \xi(t) a \xi(t-d(t)) \nonumber \\
										\leq& r^2 \xi(t)^2 + a^2 \xi(t-d(t))^2 \nonumber \\
										\leq&	r V(t,\xi(t)) + \lambda_1 V(t-d(t),\xi(t-d(t))),	\label{eq:dini bound}
	\end{align}
	along the solutions of (\ref{eq:LTI}) with $\tilde{u},\tilde{\nu} \equiv 0$, for each $t \not \in \mathcal{T}$.	In what follows, we are going to show that
	\begin{align}
		V(t,\xi(t)) \leq r M \|\xi_{t_0}\|^2 e^{-\lambda (t-t_0)}, \qquad \forall t \geq t_0,		\label{eq:show this}
	\end{align}
	where
	\begin{align}	
		M > e^{\lambda \tau} \geq e^{\lambda (t_1 - t_0)}.	\label{eq:M}	
	\end{align}
	One can easily verify that (I) implies (\ref{eq:M}). Notice that this choice of $M$ yields $\|\xi_{t_0}\|^2 < M \|\xi_{t_0}\|^2 e^{-\lambda (t_1 - t_0)}$.
	
	According to the principle of mathematical induction, we start showing that
	\begin{align}
		V(t,\xi(t)) \leq r M \|\xi_{t_0}\|^2 e^{-\lambda (t-t_0)}, \qquad \forall t \in [t_0,t_1),		\label{eq:first interval}
	\end{align}
	holds by showing that the basis of mathematical induction
	\begin{align}
		V(t,\xi(t)) \leq r M \|\xi_{t_0}\|^2 e^{-\lambda (t_1-t_0)}, \qquad \forall t \in [t_0,t_1),		\label{eq:first interval strict}
	\end{align}
	holds. For the sake of contradiction, suppose that (\ref{eq:first interval strict}) does not hold. From (\ref{eq:lyap bounds}) and (\ref{eq:M}), we infer that there exists $\overline{t} \in (t_0,t_1)$ such that
	\begin{align*}
		V(\overline{t},\xi(\overline{t})) &> r M \|\xi_{t_0}\|^2 e^{-\lambda (t_1-t_0)}	\nonumber	\\
																			&> r \|\xi_{t_0}\|^2 \geq V(t_0+s,\xi(t_0+s)), \;\; \forall s \in [-\breve{d},0],
	\end{align*}
	which implies that there exists $t^* \in (t_0,\overline{t})$ such that
	\begin{align}
		V(t^*,\xi(t^*)) &= r M \|\xi_{t_0}\|^2 e^{-\lambda (t_1-t_0)},	\nonumber \\
		V(t,\xi(t)) &\leq V(t^*,\xi(t^*)),	\qquad \forall t \in [t_0-\breve{d},t^*],		\label{eq:star}
	\end{align}
	and there exists $t^{**} \in [t_0,t^*)$ such that
	\begin{align}
		V(t^{**},\xi(t^{**})) &= r \|\xi_{t_0}\|^2,	\nonumber \\
		V(t^{**},\xi(t^{**})) &\leq V(t,\xi(t)),	\qquad \qquad	 \forall t \in [t^{**},t^*].		\label{eq:starstar}
	\end{align}
	Using (\ref{eq:star}) and (\ref{eq:starstar}), for any $s \in [-\breve{d},0]$ we have
	\begin{align}
		V(t+s&,\xi(t+s)) \leq r M \|\xi_{t_0}\|^2 e^{-\lambda (t_1-t_0)}			\nonumber \\	
										&= M e^{-\lambda (t_1-t_0)} V(t^{**},\xi(t^{**})) 		\nonumber \\
										&\leq M e^{-\lambda (t_1-t_0)} V(t,\xi(t)), \quad \forall t \in [t^{**},t^*].		\label{eq:contradiction 1}
	\end{align}
	Let us now take $s$ to be a function of time, that is, $s:=s(t)$. From (\ref{eq:dini bound}) and (\ref{eq:contradiction 1}) with $s:=s(t)=-d(t) \in [-\breve{d},0]$ for all $t \in [t^{**},t^*]$, we obtain
	\begin{align*}
		D^+ V(t,\xi(t)) \leq \big( r + \lambda_1 M e^{-\lambda (t_1-t_0)} \big) V(t,\xi(t)),
	\end{align*}
	for all $t \in [t^{**},t^*]$. Recall that $0 < t^*-t^{**} \leq t_1-t_0 \leq \tau$. Having that said, it follows from (\ref{eq:M}), (\ref{eq:star}) and (\ref{eq:starstar}) that
	\begin{align*}
		V(t^*,\xi(t^*)) &\leq V(t^{**},\xi(t^{**})) e^{( r + \lambda_1 M e^{-\lambda (t_1-t_0)} ) (t^*-t^{**})}	\nonumber \\
										&= r \|\xi_{t_0}\|^2 e^{( r + \lambda_1 M e^{-\lambda (t_1-t_0)} ) (t^*-t^{**})}	\nonumber \\
										&\leq r \|\xi_{t_0}\|^2 e^{( r + \lambda_1 M e^{-\lambda \tau} ) \tau}	\overset{\textrm{(I)}}{<} r M	\|\xi_{t_0}\|^2 e^{-\lambda \tau} \nonumber \\
										&\leq r M	\|\xi_{t_0}\|^2 e^{-\lambda(t_1-t_0)} = V(t^*,\xi(t^*)),
	\end{align*}
	which is a contradiction. Hence, (\ref{eq:first interval}) holds, i.e., (\ref{eq:show this}) holds over $[t_0,t_1)$.
	
	It is now left to show that (\ref{eq:show this}) holds over $[t_{k-1},t_{k})$ for each $k \in \mathbb{N}$, $k \geq 2$. To that end, assume that (\ref{eq:show this}) holds for each $k \in \{1,2,\ldots,m\}$, where $m \in \mathbb{N}$, i.e.,
	\begin{align}
		V(t,\xi(t)) \leq r M \|\xi_{t_0}\|^2 e^{-\lambda (t-t_0)}, \qquad \forall t \in [t_0,t_k), 	\label{eq:induction assm}
	\end{align}
	for every $k \in \{1,2,\ldots,m\}$. Let us now show that (\ref{eq:show this}) holds over $[t_m,t_{m+1})$ as well, i.e.,
	\begin{align}
		V(t,\xi(t)) \leq r M \|\xi_{t_0}\|^2 e^{-\lambda (t-t_0)}, \;\; \forall t \in [t_m,t_{m+1}).	\label{eq:induction step}
	\end{align}
	For the sake of contradiction, suppose that (\ref{eq:induction step}) does not hold. Then, we can define
	\begin{align*}
		\overline{t} := \inf \Big\{ t \in [t_m, t_{m+1}) \big| V (t, \xi(t)) > r M \|\xi_{t_0}\|^2 e^{-\lambda (t-t_0)} \Big\}.
	\end{align*}
	From (\ref{eq:LTI jump}) with $\tilde{\nu} \equiv 0$ and (\ref{eq:induction assm}), we know that 
		\begin{align*}
			V(t_m^+,\xi(t_m^+)) &= r \xi(t_m^+)^2 = r c^2 \xi(t_m)^2 = \lambda_2 V(t_m^-,\xi(t_m^-))	\nonumber \\
													&\leq \lambda_2 r M \| \xi_{t_0} \|^2 e^{-\lambda (t_m-t_0)} \nonumber \\
													&= \lambda_2 r M \| \xi_{t_0} \|^2 e^{\lambda (\overline{t}-t_m)} e^{-\lambda (\overline{t}-t_0)}	\nonumber \\
													&< \lambda_2 r e^{\lambda (t_{m+1}-t_m)} M \| \xi_{t_0} \|^2 e^{-\lambda (\overline{t}-t_0)}	\nonumber \\
													&< r M \| \xi_{t_0} \|^2 e^{-\lambda (\overline{t}-t_0)},
		\end{align*}
	where $\lambda_2 \in (0,1)$ is such that $\lambda_2 e^{\lambda(t_{m+1}-t_m)} \leq \lambda_2 e^{\lambda \tau} < 1$. One can easily verify that (II) implies $\lambda_2 e^{\lambda \tau} < 1$. Another fact to notice is that $\overline{t} \neq t_m$. Employing the continuity of $V(t,\xi(t))$ over the interval $[t_m,t_{m+1})$, we infer
	\begin{align}
		V (\overline{t}, \xi(\overline{t})) &= r M \|\xi_{t_0}\|^2 e^{-\lambda (\overline{t}-t_0)},	\nonumber \\
		V (t, \xi(t)) &\leq V (\overline{t}, \xi(\overline{t})), \qquad \qquad	\forall t \in [t_m,\overline{t}].			\label{eq:overlinet}
	\end{align}
	In addition, we know that there exists $t^* \in (t_m,\overline{t})$ such that
	\begin{align}
		V (t^*,\xi(t^*)) &= \lambda_2 r e^{\lambda(t_{m+1}-t_m)} M \|\xi_{t_0}\|^2 e^{-\lambda(\overline{t}-t_0)},		\nonumber \\
		V (t^*,\xi(t^*)) &\leq V (t, \xi(t)) \leq V (\overline{t},\xi(\overline{t})), \quad	\forall t \in [t^*,\overline{t}].		\label{eq:sandwich}
	\end{align}
	
	We proceed as follows, for any $t \in [t^*, \overline{t}]$ and any $s \in [-\breve{d}, 0]$, then either $t + s \in [t_0 - \breve{d}, t_m)$ or $t + s \in [t_m, \overline{t}]$. If $t + s \in [t_0 - \breve{d}, t_m)$, then from (\ref{eq:induction assm}) we have
	\begin{align}
		V(t+s,&\xi(t+s)) \leq r M \|\xi_{t_0}\|^2 e^{-\lambda (t+s-t_0)} \nonumber \\
										&= r M \|\xi_{t_0}\|^2 e^{-\lambda (t-t_0)} e^{-\lambda s}		\nonumber \\
									  &\leq r M \|\xi_{t_0}\|^2 e^{-\lambda (\overline{t} - t_0)} e^{\lambda (\overline{t}-t)} e^{\lambda \breve{d}}	\nonumber \\
										&\leq r e^{\lambda \breve{d}} e^{\lambda (t_{m+1}-t_m)} M \|\xi_{t_0}\|^2 e^{-\lambda (\overline{t} - t_0)}.		\label{eq:case 1}
	\end{align}
	If $t + s \in [t_m, \overline{t}]$, then from (\ref{eq:overlinet}) we have
	\begin{align}
		V(t+s,\xi(&t+s)) \leq r M \|\xi_{t_0}\|^2 e^{-\lambda (\overline{t}-t_0)}	\nonumber \\
										 &\leq r e^{\lambda \breve{d}} e^{\lambda (t_{m+1}-t_m)} M \|\xi_{t_0}\|^2 e^{-\lambda (\overline{t} - t_0)}.		\label{eq:case 2}
	\end{align}
	Apparently, the upper bounds (\ref{eq:case 1}) and (\ref{eq:case 2}) are the same; hence, it does not matter whether $t + s \in [t_0 - \breve{d}, t_m)$ or $t + s \in [t_m, \overline{t}]$. Therefore, from (\ref{eq:sandwich}) and this upper bound, we have for any $s \in [-\breve{d}, 0]$
	\begin{align}
		V(t+s,\xi(t+s)) \leq \frac{e^{\lambda \breve{d}}}{\lambda_2} V(t^*,\xi(t^*)) \leq \frac{e^{\lambda \breve{d}}}{\lambda_2} V(t,\xi(t)),	\label{eq:dini 2}
	\end{align}
	for all $t \in [t^*, \overline{t}]$. Once more, let us take $s$ to be a function of time, that is, $s:=s(t)=-d(t) \in [-\breve{d},0]$ for all $t \in [t^*, \overline{t}]$. Now, from (\ref{eq:dini bound}) and (\ref{eq:dini 2}), we have
	\begin{align*}
		D^+ V(t,\xi(t)) \leq \big( r + \frac{\lambda_1}{\lambda_2} e^{\lambda \breve{d}} \big) V(t,\xi(t)), \qquad \forall t \in [t^*, \overline{t}].
	\end{align*}
	Recall that $0 < \overline{t}-t^{*} \leq t_{m+1}-t_m \leq \tau$. Accordingly, we reach
	\begin{align*}
		V(&\overline{t},\xi(\overline{t})) \leq V(t^*,\xi(t^*)) e^{( r + \frac{\lambda_1}{\lambda_2} e^{\lambda \breve{d}} ) (\overline{t}-t^*)} \nonumber \\
																			&= \lambda_2 r e^{\lambda(t_{m+1}-t_m)} M \|\xi_{t_0}\|^2 e^{-\lambda(\overline{t}-t_0)} e^{( r + \frac{\lambda_1}{\lambda_2} e^{\lambda \breve{d}} ) (\overline{t}-t^*)}		\nonumber \\
																		  &\leq \lambda_2 e^{\lambda \tau} e^{( r + \frac{\lambda_1}{\lambda_2} e^{\lambda \breve{d}} ) \tau} r M \|\xi_{t_0}\|^2 e^{-\lambda (\overline{t}-t_0)}		\nonumber \\
																			&\overset{\textrm{(II)}}{<} r M \|\xi_{t_0}\|^2 e^{-\lambda (\overline{t}-t_0)} = V (\overline{t}, \xi(\overline{t})),
	\end{align*}
	which is a contradiction; hence, (\ref{eq:show this}) holds over $[t_m,t_{m+1})$. Employing mathematical induction, one immediately infers that (\ref{eq:show this}) holds over $[t_{k-1},t_{k})$ for each $k \in \mathbb{N}$. From (\ref{eq:lyap bounds}) and (\ref{eq:show this}), it follows that
	\begin{align*}
		\|\xi(t)\| \leq \sqrt{M} \|\xi_{t_0}\| e^{-\frac{\lambda}{2} (t-t_0)}, \qquad \forall t \geq t_0.
	\end{align*}
\end{proof}

\subsection{Proof of Theorem \ref{tm:Lp}}

The following two well-known results can be found in, for example, \cite{mtabbara2007}.

\begin{lemma} [Young's Inequality]		\label{lemma:young}
		Let $*$ denote convolution over an interval $I$, $f \in \mathcal{L}_p[I]$ and $g \in \mathcal{L}_q[I]$. The Young's inequality is $\|f*g\|_r \leq \|f\|_p \|g\|_q$ for $\frac{1}{r} = \frac{1}{p} + \frac{1}{q} - 1$ where $p,q,r > 0$.
\end{lemma}

\begin{theorem} [Riesz-Thorin Interpolation Theorem]		\label{thm:riesz}
		Let $F: \mathbb{R}^{n} \rightarrow \mathbb{R}^{m}$ be a linear operator and suppose that $p_0, p_1, q_0, q_1 \in [1,\infty]$ satisfy $p_0 < p_1$ and $q_0 < q_1$. For any $\theta \in [0,1]$ define $p_{\theta}, q_{\theta}$ by $1/p_{\theta} = (1-{\theta})/p_0 + {\theta}/p_1$ and $1/q_{\theta} = (1-{\theta}/q_0) + {\theta}/q_1$. Then, $\|F\|_{p_{\theta} \rightarrow q_{\theta}} \leq \|F\|^{1-{\theta}}_{p_0 \rightarrow q_0} \|F\|^{{\theta}}_{p_1 \rightarrow q_1}$, where $\|F\|_{p \cdot \rightarrow q \cdot}$ denotes the norm of the mapping $F$ between the $\mathcal{L}_{p \cdot}$ and $\mathcal{L}_{q \cdot}$ space. In particular, if $\|F\|_{p_0 \rightarrow q_0} \leq M_0$ and $\|F\|_{p_1 \rightarrow q_1} \leq M_1$, then $\|F\|_{p_{\theta} \rightarrow q_{\theta}} \leq M_0^{1-{\theta}} M_1^{\theta}$.
\end{theorem}

\begin{proof}
	From the UGES assumption of the theorem, we infer that the fundamental matrix $\Phi(t,t_0)$ satisfies:
	\begin{align*}
		\|\Phi(t,t_0)\| \leq  \sqrt{M} e^{-\frac{\lambda}{2} (t-t_0)}, \qquad \forall t \geq t_0,
	\end{align*}
	uniformly in $t_0$. Refer to \cite[Definition 3.]{aanokhin1995} for the exact definition of a fundamental matrix. Now, \cite[Theorem 3.1.]{aanokhin1995} provides
	\begin{align}
		\xi(t) &= \Phi(t,t_0) \xi(t_0) + \int_{t_0}^{t} \Phi(t,s) \tilde{u}(s) \mathrm{d}s \nonumber \\
					 &+ \int_{t_0}^{t} \Phi(t,s) 	a \xi_s(-d(t)) \mathrm{d}s \nonumber \\
					 &+ \sum_{t_0 < t_i \leq t} \Phi(t,t_i) \tilde{\nu}(t_i), \qquad \qquad \qquad \forall t \geq t_0,	\label{eq:linear op riesz}
	\end{align}
	where $\xi_s(-d(t))=0$ when $s-d(t) \geq t_0$. The above equality along with
	$$\sum_{t_0 < t_i \leq t} e^{-\frac{\lambda}{2} (t-t_i)} \leq \sum_{i=1}^{\infty} e^{-\frac{\lambda}{2} \varepsilon i}  = \frac{1}{e^{\frac{\lambda \varepsilon}{2}} - 1}$$
	immediately yields
	\begin{align}
		\|\xi(t)&\| \leq \|\Phi(t,t_0)\| \|\xi_{t_0}\| + \int_{t_0}^{t} \|\Phi(t,s)\| \|\tilde{u}(s)\| \mathrm{d}s \nonumber \\
								&+ \int_{t_0}^{t} \|\Phi(t,s)\| |a| \|\xi_s(-d(t))\| \mathrm{d}s \nonumber \\
								&+ \sum_{t_0 < t_i \leq t} \|\Phi(t,t_i)\| \|\tilde{\nu}(t_i)\|		\nonumber \\
						    &\leq \sqrt{M} e^{-\frac{\lambda}{2} (t-t_0)} \|\xi_{t_0}\| + \sqrt{M} \int_{t_0}^{t} e^{-\frac{\lambda}{2} (t-s)} \|\tilde{u}(s)\| \mathrm{d}s \nonumber \\
						    &+ |a| \sqrt{M} \|\xi_{t_0}\| \int_{t_0}^{t_0 + d(t)}  e^{-\frac{\lambda}{2} (t-s)} \mathrm{d}s + \nonumber \\
						    &+ \tilde{K}_{\nu} \sqrt{M} \frac{1}{e^{\frac{\lambda \varepsilon}{2}} - 1},  \qquad \qquad \qquad  \forall t \geq t_0, \label{eq:convolution}
	\end{align}
	uniformly in $t_0$.
	
	Let us now estimate the contribution of the initial condition $\xi_{t_0}$ towards $\|\xi[t_0,t]\|_p$ by setting $\tilde{u} \equiv 0$ and $\tilde{K}_{\nu}=0$. In other words, we have
	\begin{align}
		\|\xi(t)\| &\leq  \sqrt{M} e^{-\frac{\lambda}{2} (t-t_0)} \|\xi_{t_0}\| \nonumber \\
							 &+	|a| \sqrt{M} \|\xi_{t_0}\| \int_{t_0}^{t_0 + \breve{d}}  e^{-\frac{\lambda}{2} (t-s)} \mathrm{d}s, \quad \forall t \geq t_0.		\label{eq:const init}
	\end{align}
	In what follows, we use $(a+b)^p \leq 2^{p-1} a^p + 2^{p-1} b^p$ and $(a+b)^{\frac{1}{p}} \leq a^{\frac{1}{p}} + b^{\frac{1}{p}}$, where $a,b\geq 0$ and $p \in [1, \infty)$ (see, for example, \cite[Lemma 1 \& 2]{rfreeman2004}). Raising (\ref{eq:const init}) to the $p^{\mathrm{th}} \in [1,\infty)$ power, integrating over $[t_{0},t]$ and taking the $p^{\mathrm{th}}$ root yields
	\begin{align*}
		\|\xi[t_0,t]\|_p		&\leq \Bigg( \sqrt{M}  +  |a| \sqrt{M} \frac{2}{\lambda} \big(e^{\frac{\breve{d} \lambda}{2}} - 1 \big) \Bigg) 2^{\frac{p-1}{p}} \times \nonumber \\
												&\times		\Big( \frac{2}{p \lambda} \Big)^{\frac{1}{p}} \|\xi_{t_0}\|, \qquad \qquad \qquad \forall t \geq t_0,
	\end{align*}
	where we used
	\begin{align}
		\Bigg( \int_{t_0}^{\infty} \Big( \int_{t_0}^{t_0 + \breve{d}}  e^{-\frac{\lambda}{2} (t-s)} \mathrm{d}s \Big)^{p} \mathrm{d}t \Bigg)^{\frac{1}{p}} = \frac{2}{\lambda} \big(e^{\frac{\breve{d} \lambda}{2}} - 1 \big) \Big( \frac{2}{p \lambda} \Big)^{\frac{1}{p}}.
	\end{align}
	When $p=\infty$, simply take the limit $$\lim_{p \rightarrow \infty} 2^{\frac{p-1}{p}} \Big( \frac{2}{p \lambda} \Big)^{\frac{1}{p}} = 2.$$
	
	Let us now estimate the contribution of the input $\tilde{u}(t)$ towards $\|\xi[t_0,t]\|_p$ by setting $\|\xi_{t_0}\| = 0$ and $\tilde{K}_{\nu}=0$. In other words, we have
	\begin{align}
		\|\xi(t)\| \leq \sqrt{M} \int_{t_0}^{t} e^{-\frac{\lambda}{2} (t-s)} \|\tilde{u}(s)\| \mathrm{d}s,  \qquad \forall t \geq t_0. \label{eq:y contrib}
	\end{align}
	Note that $\int_{0}^{\infty}  e^{-\frac{\lambda}{2}s} \mathrm{d}s = \frac{2}{\lambda}$. Now, integrating the previous inequality over $[t_0,t]$ and using Lemma \ref{lemma:young} with $p=q=r=1$ yields the $\mathcal{L}_1$-norm estimate:
	\begin{align}
		\| \xi[t_{0},t] \|_1 \leq  \frac{2}{\lambda} \sqrt{M} \|\tilde{u} [t_{0},t] \|_1,  \qquad \qquad \forall t \geq t_0.		\label{eq:L1}
	\end{align}
	Taking the $\max$ over $[t_0,t]$ in (\ref{eq:y contrib}) and using Lemma \ref{lemma:young} with $q=r=\infty$ and $p=1$ yields the $\mathcal{L}_{\infty}$-norm estimate:
	\begin{align}
		\| \xi[t_{0},t] \|_{\infty} \leq  \frac{2}{\lambda} \sqrt{M} \|\tilde{u} [t_{0},t] \|_{\infty},  \qquad \quad \forall t \geq t_0.		\label{eq:Linfty}
	\end{align}
	From (\ref{eq:linear op riesz}), one infers that we are dealing with a linear operator, say $F$, that maps $\tilde{u}$ to $e$ with bounds for the norms $\|F\|_1 \leq \|F\|_1^*$ and $\|F\|_{\infty} \leq \|F\|_{\infty}^*$, where $\|F\|_1^*$ and $\|F\|_{\infty}^*$ are given by (\ref{eq:L1}) and (\ref{eq:Linfty}), respectively. Because $\|F\|_1^* = \|F\|_{\infty}^*$, Theorem \ref{thm:riesz} gives that $\|F \|_p \leq \|F\|^*_1=\|F\|_{\infty}^*$ for all $p \in [1,\infty]$. This yields
		\begin{align}
		\| \xi[t_{0},t] \|_{p} \leq  \frac{2}{\lambda} \sqrt{M} \|\tilde{u} [t_{0},t] \|_{p},  \qquad \qquad \forall t \geq t_0,		\label{eq:Lp}
	\end{align}
	for any $p \in [1,\infty]$.

	Let us now estimate the contribution of the noise $\tilde{\nu}(t)$ towards $\|\xi[t_0,t]\|_p$ by setting $\|\xi_{t_0}\| = 0$ and $\tilde{u} \equiv 0$. In other words, we have
	\begin{align*}
		\|\xi(t)\| \leq \tilde{K}_{\nu} \sqrt{M} \frac{1}{e^{\frac{\lambda \varepsilon}{2}} - 1},  \qquad \qquad \forall t \geq t_0.
	\end{align*}
	By identifying $b(t) \equiv b := \tilde{K}_{\nu} \sqrt{M} \frac{1}{e^{\frac{\lambda \varepsilon}{2}} - 1}$, we immediately obtain
	\begin{align*}
		\|\xi[t_0,t]\|_p \leq \|b[t_0,t]\|_p,  \qquad \qquad \qquad \forall t \geq t_0,
	\end{align*}
	for all $p \in [1,\infty]$.
	
	Finally, summing up the contributions of $\xi_{t_0}$, $\tilde{u}(t)$ and $\tilde{\nu}(t)$ produces
	\begin{align*}
		\|\xi[t_0,t]\|_p  &\leq 2 \sqrt{M} \Bigg( 1 + |a| \frac{2}{\lambda} \big(e^{\frac{\breve{d} \lambda}{2}} - 1 \big) \Bigg) \Big( \frac{1}{p \lambda} \Big)^{\frac{1}{p}} \|\xi_{t_0}\| \nonumber \\
											&+ \frac{2}{\lambda} \sqrt{M} \|\tilde{u} [t_{0},t] \|_{p} + \|b[t_0,t]\|_p, \qquad \forall t \geq t_0,
	\end{align*}
	for any $p \in [1,\infty]$.
\end{proof}

\subsection{Proof of Theorem \ref{thm:main}}

\begin{proof}
	Combining (i) of UGES protocols and (\ref{eq:complex}), one obtains:
	\begin{align}
		\Big\langle \frac{\partial W(i,e)}{\partial e}, g(t,x_t,e_t,\omega_t) \Big\rangle \leq& \frac{\overline{a}}{\underline{a}} L W(j,e(t-d(t))) \nonumber \\
																																												  &+ \|H(x_t,\omega_t)\|.	\label{eq:complex 2}
	\end{align}
	for any $i,j \in \mathbb{N}$. Hence, the index $i$ in $W(i,e)$ can be omitted in what follows. Now, we define $Z(t):=W(e(t))$ and reach
	\begin{align}
		\frac{d Z(t)}{d t} &\leq \frac{\overline{a}}{\underline{a}} L Z(t-d(t)) + \|H(x_t,\omega_t)\|,	\label{eq:w flow}
	\end{align}
	for almost all $t \notin \mathcal{T}$. For a justification of the transition from (\ref{eq:complex 2}) to (\ref{eq:w flow}), refer to \cite[Footnote 8]{dnesic2004}. Likewise, property (ii) of UGES protocols yields
	\begin{align}
		Z(t^+) \leq \rho Z(t) + \overline{a} \nu_j(t),	\label{eq:w jump}
	\end{align}
	for all $t \in \mathcal{T}$, where $\nu_j(t)$, $j \in \{1,\ldots,l\}$, is the $j^{\mathrm{th}}$ NCS link noise given by (\ref{eq:noise}) and upper bounded with $K_{\nu}$. Notice that $|Z(t)|=|W(e(t))|$. Next, we use the comparison lemma for impulsive delayed systems \cite[Lemma 2.2]{xliu2005}. Basically, the fundamental matrix of (\ref{eq:w flow})-(\ref{eq:w jump}) is upper bounded with the fundamental matrix of (\ref{eq:LTI}) with parameters $a:=\frac{\overline{a}}{\underline{a}} L$ and $c:=\rho$. Refer to \cite[Definition 3.]{aanokhin1995} for the exact definition of a fundamental matrix. Of course, the corresponding transmission interval $\tau$ in (\ref{eq:LTI jump}), and therefore in (\ref{eq:w jump}), has to allow for $M>1$ and $\lambda > 0$ that satisfy (I), (II) and $\frac{2}{\lambda} \sqrt{M} \gamma_H < 1$ (as stated in Theorem \ref{thm:main}). Essentially, (I) and (II) yield $\mathcal{L}_p$-stability with bias from $H$ to $W$, while $\frac{2}{\lambda} \sqrt{M} \gamma_H < 1$ allows us to invoke the small-gain theorem. Following the proof of Theorem \ref{tm:Lp}, one readily establishes $\mathcal{L}_p$-stability from $H$ to $W$ with bias, i.e.,
	\begin{align}
			\|W[t_0,t]\|_p \leq K_W \|W_{t_0}\| + \gamma_W \|H [t_{0},t] \|_{p} + \|b[t_0,t]\|_p,		\label{eq:Lp W}
	\end{align}
	for any $t \geq t_0$ any $p \in [1,\infty]$, where 
	\begin{align*}
		K_W &:= 2 \sqrt{M} \Bigg( 1 + \frac{\overline{a}}{\underline{a}} L \frac{2}{\lambda} \big(e^{\frac{\breve{d} \lambda}{2}} - 1 \big) \Bigg) \Big( \frac{1}{p \lambda} \Big)^{\frac{1}{p}}, \nonumber \\
		\gamma_W &:= \frac{2}{\lambda} \sqrt{M},	\quad b := \frac{\overline{a} K_{\nu} \sqrt{M}}{e^{\frac{\lambda \varepsilon}{2}} - 1}.
	\end{align*}
	
	Let us now infer $\mathcal{L}_p$-stability with bias from $\omega$ to $(H,e)$ via the small-gain theorem. Inequality (\ref{eq:Lp H}) implies
	\begin{align}
		\|H[t,t_0]\|_p \leq& K_H \|x_{t_0}\| + \gamma_H \|W[t,t_0]\|_p + \gamma_H \|\omega[t_0,t]\|_p,		\label{eq:Lp H 2}
	\end{align}
	for all $t \geq t_0$. Combining the above with (\ref{eq:Lp W}) and property (i) of UGES protocols yields
	\begin{align}
		&\|e[t_0,t]\|_p \leq \frac{\overline{a} K_W / \underline{a} }{(1-\gamma_W \gamma_H)}\|e_{t_0}\| + \frac{\gamma_W K_H / \underline{a}}{(1-\gamma_W \gamma_H)} \|x_{t_0}\| \nonumber \\
		&+ \frac{\gamma_W \gamma_H / \underline{a}}{(1-\gamma_W \gamma_H)} \|\omega[t_0,t]\|_p + \frac{1/ \underline{a}}{(1-\gamma_W \gamma_H)} \|b[t_0,t]\|_p, \label{eq:Lp e} 
	\end{align}
	\begin{align}
		&\|H[t_0,t]\|_p \leq \frac{K_H}{1-\gamma_W \gamma_H} \|x_{t_0}\| + \frac{\overline{a} K_H K_W}{1-\gamma_W \gamma_H} \|e_{t_0}\| \nonumber \\
		&+ \frac{\gamma_H}{1-\gamma_W \gamma_H} \|\omega[t_0,t\|_p + \frac{\gamma_H}{1-\gamma_W \gamma_H} \|b[t_0,t\|_p.	\nonumber
	\end{align}
	From the above two inequalities, $\mathcal{L}_p$-stability from $\omega$ to $(H,e)$ with bias $\frac{\gamma_H + \frac{1}{\underline{a}}}{1-\gamma_W \gamma_H} \frac{\overline{a} K_{\nu} \sqrt{M}}{e^{\frac{\lambda \varepsilon}{2}} - 1}$ and gain $\frac{\gamma_H (1+\frac{\gamma_W}{\underline{a}})}{1-\gamma_W \gamma_H}$ is immediately obtained.
\end{proof}

\subsection{Proof of Corollary \ref{cor:main}}

\begin{proof}
	The $\mathcal{L}_p$-detectability of $x$ from $(W,\omega,H)$ implies that there exist $K_d,\gamma_d \geq 0$ such that
	\begin{align}
		\|x[t_{0},t]\|_p \leq& K_d \|x_{t_0}\| + \gamma_d \|H[t_{0},t]\|_p + \gamma_d \|(W,\omega)[t_{0},t]\|_p	\nonumber \\
										 \leq&  K_d \|x_{t_0}\| + \gamma_d \|H[t_{0},t]\|_p + \gamma_d \|W[t_{0},t]\|_p \nonumber \\
										  +& \gamma_d \|\omega[t_{0},t]\|_p 			\label{eq:Lp x}	
	\end{align}
	for all $t \geq t_{0}$. Plugging (\ref{eq:Lp H 2}) into (\ref{eq:Lp x}) leads to
	\begin{align*}
		\|x[&t_{0},t]\|_p \leq K_d \|x_{t_0}\| + \gamma_d K_H \|x_{t_0}\| + \gamma_d \gamma_H \|W[t,t_0]\|_p \nonumber \\
				&+ \gamma_d \gamma_H \|\omega[t_0,t]\|_p + \gamma_d \|W[t_{0},t]\|_p + \gamma_d \|\omega[t_{0},t]\|_p \nonumber \\
				&\leq K_d \|x_{t_0}\| + \gamma_d K_H \|x_{t_0}\| + (\overline{a} \gamma_d \gamma_H + \overline{a} \gamma_d) \|e[t,t_0]\|_p \nonumber \\
				&+ (\gamma_d \gamma_H + \gamma_d) \|\omega[t_{0},t]\|_p
	\end{align*}
	for all $t \geq t_{0}$. Finally, we include (\ref{eq:Lp e}) into (\ref{eq:Lp x}) and add the obtained inequality to (\ref{eq:Lp e}) which establishes $\mathcal{L}_p$-stability with bias from $\omega$ to $(x,e)$.
\end{proof}

\subsection{Proof of Proposition \ref{prop:UGAS} }

\begin{proof}
	For the case $p=\infty$, UGS of the interconnection $\Sigma_n$ and $\Sigma_e$ is immediately obtained using the definition of $\mathcal{L}_{\infty}$-norm. Therefore, the case $p \in [1,\infty)$ is more interesting. From the conditions of the proposition, we know that there exist $K \geq 0$ and $\gamma \geq 0$ such that
	\begin{align*}
		\|(x,e)[t_0,t]\|_p \leq K \|(x_{t_0},e_{t_0})\| + \gamma \|\omega[t_0,t]\|_p,	\;	\forall t \geq t_0.
	\end{align*}
	Recall that $\omega \equiv\textbf{0}_{n_{\omega}}$ when one is interested in asymptotic stability. By raising both sides of the above inequality to the $p^{\mathrm{th}}$ power, we obtain:
	\begin{align}
		\int_{t_0}^t \|(x,e)(s)\|^p \mathrm{d}s \leq K_1 \|(x_{t_0},e_{t_0})\|^p, \qquad \forall t \geq t_0,		\label{eq:for UGAS}
	\end{align}
	where $K_1:=K^p$.
	
	First, we need to establish UGS of the interconnection $\Sigma_n$ and $\Sigma_e$ when $\omega \equiv\textbf{0}_{n_{\omega}}$. Before we continue, note that the jumps in (\ref{eq:Sigman}) and (\ref{eq:Sigmae}) are such that $\|(x,e)(t^+)\| \leq \|(x,e)(t)\|$ for each $t \in \mathcal{T}$. Apparently, jumps are not destabilizing and can be disregarded in what follows. Along the lines of the proof for \cite[Theorem 1]{edsontag1998}, we pick any $\epsilon > 0$. Let $\mathcal{K}$ be the set $\{ (x,e) \big| \frac{\epsilon}{2} \leq \|(x,e)\| \leq \epsilon \}$, $\mathcal{K}_1$ be the set $\{ (x,e) \big| \|(x,e)\| \leq \epsilon \}$, and take
$$a:=\sup_{ \substack{t \in \mathbb{R}\\ (x,e)(t) \in \mathcal{K}\\ (x,e)_t \in PC([-\overline{d},0],\mathcal{K}_1)}} \|\big( f(t,x_t,e,\textbf{0}_{n_{\omega}}),g(t,x_t,e_t,\tilde{\textbf{0}}_{n_{\omega}}) \big)\|,$$
where $f(t,x_t,e,\mathbf{0}_{n_{\omega}})$ and $g(t,x_t,e_t,\tilde{\mathbf{0}}_{n_{\omega}})$ are given by (\ref{eq:nesic 1}) and (\ref{eq:nesic 2}), respectively. This supremum exists because the underlying dynamics are Lipschitz uniformly in $t$. Next, choose $0 < \delta < \frac{\epsilon}{2}$ such that $r < \delta$ implies $K_1 \|(x_{t_0},e_{t_0})\|^p < s_0$, where $s_0 := \frac{\epsilon (\epsilon/2)^p}{2a}$ . Let $\|(x_{t_0},e_{t_0})\|_d < \delta$. Then $\|(x,e)(t)\| < \epsilon$ for all $t \geq t_0$. Indeed, suppose that there exists some $t > t_0$ such that $\|(x,e)(t)\| \geq \epsilon$. Then, there is an interval $[t_1, t_2]$ such that $\|(x,e)(t_1)\|=\frac{\epsilon}{2}$, $\|(x,e)(t_2)\|=\epsilon$, $(x,e)(t) \in \mathcal{K}$, $(x,e)_t \in PC([-\overline{d},0],\mathcal{K}_1)$, for all $t \in [t_1,t_2]$. Hence,
	\begin{align*}
		\int_{t_0}^{\infty} \|(x,e)(s)\|^p \mathrm{d}s \geq \int_{t_1}^{t_2} \|(x,e)(s)\|^p \mathrm{d}s \geq (t_2 - t_1) \Big( \frac{\epsilon}{2} \Big)^p.
	\end{align*}
	On the other hand,
	\begin{align*}
		\frac{\epsilon}{2} &\leq \|(x,e)(t_2)-(x,e)(t_1)\| \nonumber \\
											 &\leq \int_{t_1}^{t_2}  \|(f(t,x_t,e,\textbf{0}_{n_{\omega}}),g(t,x_t,e_t,\tilde{\textbf{0}}_{n_{\omega}}))\| \mathrm{d}s \nonumber \\
											 &\leq a(t_2-t_1),
	\end{align*}
	and, combining the above with (\ref{eq:for UGAS}), we conclude that
	\begin{align*}
		s_0 \leq \int_{t_0}^{\infty} \|(x,e)(s)\|^p \mathrm{d}s \leq K_1 \|(x_{t_0},e_{t_0})\|^p,
	\end{align*}
	which is a contradiction.
	
	Second, let us show asymptotic convergence of $\|(x,e)(t)\|$ to zero. Using (\ref{eq:for UGAS}), we infer that $(x,e)(t) \in \mathcal{L}_p$. Owing to the Lipschitz dynamics of the corresponding system, we conclude that $(\dot{x},\dot{e})(t) \in \mathcal{L}_p$ as well. Now, one readily establishes asymptotic convergence of $\|x\|$ to zero using \cite[Facts 1-4]{arteel1999}. Consequently, asymptotic convergence of $\|e\|$ to zero follows from (\ref{eq:convolution}) by observing that $\xi(t)$ and $\tilde{u}(t)$ in (\ref{eq:convolution}) correspond to $W(t)$ and $H(t)$, respectively, and using (i) of Definition \ref{def:UGES pro}.
\end{proof}

\bibliographystyle{abbrv}
\bibliography{jsandra_ref}

\end{document}